\documentclass[%
 reprint,
superscriptaddress,
 amsmath,amssymb,
 aps,
]{revtex4-2}

\usepackage{graphicx}
\usepackage{dcolumn}
\usepackage{bm}
\usepackage{amsthm}
\usepackage{xcolor}
\usepackage{braket}
\usepackage{textcase}
\usepackage{arydshln}

\newcommand{\tr}{\mathrm{tr}}

\newcommand{\tp}{\mathrm{T}}
\newcommand{\iden}{\mathbb{I}}

\newcommand{\nb}{\bar{n}}

\newcommand{\Swap}{\textsc{swap}}
\newcommand{\ThEx}{\mathfrak{n}}
\newcommand{\FirM}{\mathbf{r}}
\newcommand{\SecM}{\mathbf{M}}
\newcommand{\GU}{\mathbf{G}}
\newcommand{\SU}{\mathbf{S}}
\newcommand{\CU}{\mathbf{C}}
\newcommand{\LU}{\mathbf{L}}
\newcommand{\mmZ}{\mathbf{Z}}

\newtheorem{theorem}{Theorem}

\newtheorem{definition}{Definition}
\newtheorem{lemma}{Lemma}

\newcommand{\mE}{\mathcal{E}}

\newcommand{\CQT}{Centre for Quantum Technologies, National University of Singapore, 3 Science Drive 2, Singapore 117543}
\newcommand{\NUS}{Department of Physics, National University of Singapore, 2 Science Drive 3, Singapore 117542}
\newcommand{\Shandong}{School of Information Science and Engineering, Shandong University, Qingdao 266237, China}

\begin{document}

\preprint{APS/123-QED}

\title{Sub-Bath Cooling in Bosonic Systems:\\Gaussian Constraints and Non-Gaussian Enhancements}

\author{Wen-Han Png}
\affiliation{\CQT}

\author{Xueyuan Hu}%
\email{xyhu@sdu.edu.cn}
\affiliation{\Shandong}%

\author{Valerio Scarani}
\email{physv@nus.edu.sg}
\affiliation{\CQT}
\affiliation{\NUS}

\date{\today}

\begin{abstract} 

Cooling quantum systems with finite resources is a central task in quantum technologies and has been extensively explored in discrete-variable settings. As continuous-variable (CV) platforms play an increasingly important role in quantum information processing, it becomes crucial to understand the fundamental limitations of cooling bosonic systems. In this work, we develop a general framework for cooling CV systems, identifying both the constraints imposed by Gaussianity and the advantages enabled by non-Gaussian interactions. We derive a reachable bound on the cooling performance of Gaussian operations that applies to arbitrary cooling architectures. By optimizing over all protocols saturating this bound, we further identify the most efficient scheme, which minimizes dissipated energy for a given number of ancilla modes. Beyond Gaussian operations, we show that $p$-excitation exchange exploits non-Gaussian resources to achieve a $p$-fold enhancement of the cooling limit. Our results establish the fundamental limits of CV heat-bath algorithmic cooling and reveal the crucial role of non-Gaussianity in surpassing Gaussian cooling barriers.

\end{abstract}




\maketitle

\section{\label{sec:level1} Introduction}
Continuous-variable (CV) systems are key platforms for quantum information science. They arise naturally in circuit QED, integrated photonics, optomechanics, and trapped-ion motional modes. Achieving near-ground-state preparation is crucial in CV systems, as it underpins coherence and entanglement that enable enhanced precision in quantum sensing \cite{giovannetti2011advances,toth2014quantum,demkowicz2015quantum}, robust quantum communication \cite{gisin2002quantum,pirandola2020advances}, and fault-tolerant quantum computation \cite{preskill1998fault,preskill2018quantum,fukui2018high}.
A direct route to near-ground-state preparation is achieved through cooling by contact with a cold environment. This mechanism is inherently incoherent, and millikelvin-level cooling of continuous-variable systems has been demonstrated in superconducting platforms such as transmon devices \cite{kjaergaard2020superconducting,krantz2019quantum} and optomechanical systems \cite{diamandi2025optomechanical}. Records of sub-millikelvin temperatures have been reported using sideband-resolved optomechanical cooling~\cite{teufel2011sideband}, a protocol that falls within the category of Gaussian operations. Beyond Gaussian operations, quantum absorption refrigeration based on the trilinear Hamiltonian \cite{nimmrichter2017quantum} has also been experimentally realized in a trapped-ion system \cite{maslennikov2019quantum}, demonstrating the integration of non-Gaussian operations into cooling protocols.
\begin{figure}
    \centering
    \includegraphics[width=1\linewidth]{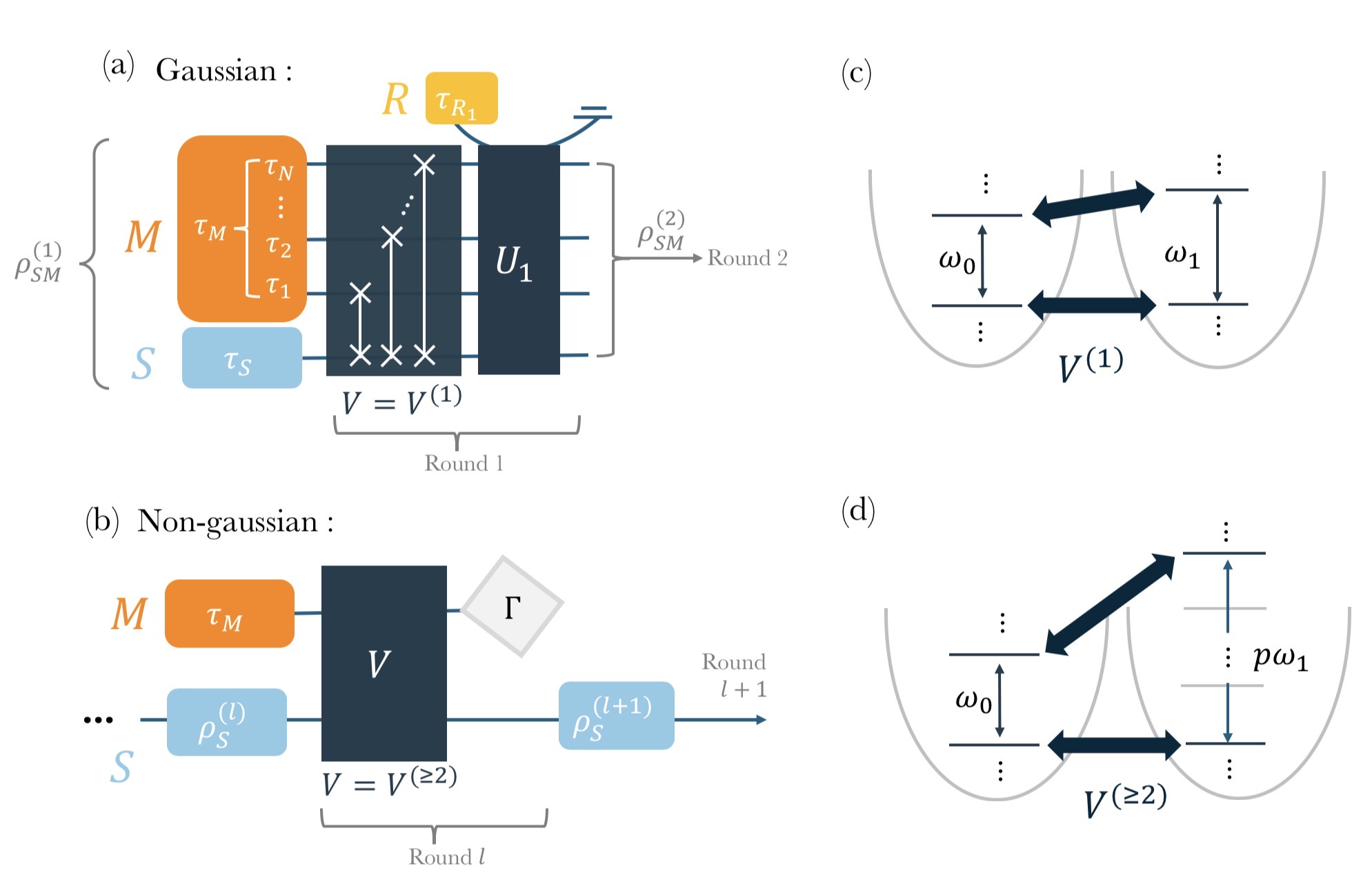}
    \caption{ Continuous-variable HBAC framework with two routines: recharging and thermalizing. $V$ denotes the recharging unitary;  $U_1$ and $\Gamma$ denote the thermalizing unitary and the full thermalization channel, respectively.
    (a) The first round of the most general Gaussian heat-bath algorithmic cooling using $N$ ancillas; The optimal recharging routine consists of successive full $\Swap$ operations, i.e., the single-excitation exchange described in (c). The output state becomes the input state of the next iteration, $\rho_{SM}^{(2)}$. (b) Example of iterative cooling with non-Gaussian operations studied in this paper. At the $l$-th recharging routine, the system interacts with a refreshed machine $\tau_M$, producing $\rho_S^{(l+1)}$ which serves as the input for the next round. The recharging routine is a $p$-excitation exchange operation for $p\geq 2$ illustrated in (d).}
    \label{fig:Overview}
\end{figure}

These experimental advances naturally raise a broader theoretical question: what are the fundamental cooling capabilities of Gaussian and non-Gaussian operations?
A notable step in this direction was made in Ref.~\cite{PhysRevLett.124.010602}, where no-cooling theorem was established for single-mode Gaussian thermal operations, even when assisted by Gaussian unitaries and catalytic states. However, the multimode generalization of the Gaussian constraint in cooling, as well as the corresponding bounds under finite-time and finite-resource constraints, are still unknown. Beyond Gaussian operations, non-Gaussian resources have been recognized as indispensable in other quantum tasks, including universal quantum computation~\cite{lloyd1999quantum}, entanglement distillation~\cite{eisert2002distilling, fiuravsek2002gaussian}, and enhanced quantum metrology~\cite{wang2025limitations,gessner2019metrological,arvidsson2020quantum}. Recent advances in the precise control of nonlinear interactions in CV platforms—such as parametric amplification and second-harmonic generation, have opened new opportunities for the integration of non-Gaussian operations into cooling protocols. Nevertheless, the impact of the nonlinearity order on achievable cooling limits has yet to be fully characterized, leaving the potential of non-Gaussian resources to enhance cooling as an open question.

Heat-bath algorithmic cooling (HBAC) has been employed as a general framework to quantitatively study the third law in finite-dimensional quantum systems, demonstrating reachable cooling limits \cite{PhysRevLett.116.170501_PPA,PhysRevLett.123.170605}, as well as the interplay between consumed resources, such as energy \cite{PRXQuantum.4.010332,PhysRevLett.134.070401}, time \cite{PRXQuantum.4.010332}, complexity \cite{PRXQuantum.4.010332,r3jn-1z54}, and memory \cite{r3jn-1z54}. It alternates between two subroutines: \emph{recharging}, where entropy is transferred from the target system to an ancilla (referred to as the machine), and \emph{thermalizing}, where this entropy is removed from the ancilla to the reservoir. The purpose of HBAC is to achieve sub-bath cooling, i.e., to cool the system mode below the reservoir temperature. Within this framework, cooling limits have been calculated for a variety of protocols, including the partner-pairing algorithm (PPA) \cite{ppa_schulman}, Nuclear Overhauser Effect (NOE) \cite{rodriguez2017heat}, state-reset (SR) \cite{Rodríguez-Briones_2017_SR}, and extended HBAC (xHBAC) \cite{Alhambra2019xHBAC}. Remarkably, a trade-off has recently been established between the dissipated heat and the size of the machine \cite{PRXQuantum.4.010332,PhysRevLett.134.070401}, bridging the theories of Landauer and Nernst. While these fundamental limits and efficiencies have been extensively studied in finite-dimensional quantum systems, they remain much less explored in continuous-variable settings. 

In this work, we systematically study the cooling limitations of Gaussian operations and the cooling enhancement of non-Gaussian operations using the HBAC framework. Our findings can be summarized as follows.
\begin{enumerate}
    \item In Sec.~\ref{sec:LimGaussianCooling}, we establish the necessary and sufficient conditions for sub-bath cooling through multimode Gaussian operations (Theorem \ref{thm:1}). 
    \item  In Sec.~\ref{sec:EnergyEfficiencyGaussian}, 
    we prove that the most efficient recharging routine corresponds to successive $\Swap$ operations on a sequence of machines ordered by their increasing energy gaps (Theorem \ref{thm:2}). Within this ordering, each energy gap remains a free parameter, which we optimize to minimize heat dissipation during Gaussian cooling.
    \item In Sec.~\ref{sec: SingleShotNonGaussian}, we establish the necessary and sufficient condition for sub-bath cooling enabled by a specific non-Gaussian operation, the $p$-excitation exchange interaction (Theorem \ref{thm:NonGCooling}). 
    Then in Sec.~\ref{sec:NonGaussianIterative}, using the HBAC framework, we prove that this iterative non-Gaussian cooling yields a $p$-fold cooling enhancement in terms of the inverse effective temperature relative to Gaussian HBAC (Theorem \ref{thm:NonGCooling2}).
\end{enumerate}

\section{Preliminary \label{sec:Preliminary}}

In this section, we establish the theoretical framework for HBAC in continuous-variable (CV) systems.
We begin by defining the two primary classes of control in the cooling protocol: coherent control and incoherent control. The distinction between the two lies in the ability to establish coherence across energy levels. While both are driven by unitary evolution, incoherent control is constrained to energy-preserving operations, whereas coherent control allows for general unitaries.

In this work, we mainly focus on the scenario based on coherent control, and will show that sub-bath cooling is unachievable with any Gaussian HBAC protocol based on incoherent control.
Precisely, we consider a $N$-mode machine $M$ that works cyclically between a target mode $S$ and a multimode reservoir $R$ to carry the entropy from $S$ to $R$ (see Fig.~\ref{fig:Overview}).
The system and machine are described through bare Hamiltonians $H_S = \omega_0 \hat{a}^\dag \hat{a}$ and $H_M = \sum_{j=1}^N \omega_j \hat{b}_j^\dag \hat b_j$ with $\omega_1\leq\cdots\leq\omega_N$, respectively. The protocol consists of several rounds, each comprising a recharging routine, in which a unitary operator $V$ acts collectively on $S$ and $M$, and a thermalizing routine, in which $SM$ interacts with $R$ via an energy-preserving unitary $U$. 
Before the protocol starts, the state of $SMR$ is initialized to a Gibbs state $\tau_S\otimes\tau_M\otimes\tau_R$ at inverse temperature $\beta$, i.e., $\tau_X=\frac{e^{-\beta H_X}}{\tr(e^{-\beta H_X})}$, $X\in\{S,M,R\}$. Throughout the manuscript, we adopt units where $\hbar = k_B = 1$. Because the reservoir is large, it is assumed that it relaxes to the Gibbs state in every round, while $SM$ can preserve its state in the output of the thermalizing routine to the next round. After the whole process ends, the state of the machine would be reset to the Gibbs state for later use.
In the Gaussian scenario (Sec.~\ref{sec:CoolingGO}), we work on the most generic Gaussian HBAC framework, i.e., the only restriction we consider is that both $V$ and $U$ are Gaussian. It leads to the fundamental restriction on the cooling limit and efficiency imposed by Gaussianity. In the non-Gaussian scenario (Sec.~\ref{sec:CoolingNGU}), we showcase the advantage of non-Gaussianity in cooling by considering a specific protocol, where $N=1$, the recharging routine is driven by the multiexcitation exchange Hamiltonian $H_I^{(p)} \propto \hat{a}\hat{b}^{\dag p} + \hat{a}^\dag\hat{b}^p$, and the thermalizing routine is simply full thermalization of the machine mode (see Fig. 1 b).
%


As a prerequisite for evaluating our cooling protocol, we formalize the definition of temperature for CV quantum systems. This is necessary because, in general, the output CV state may deviate from the standard Gibbs state. For a general single-mode state, we introduce the definitions of thermal excitation and effective temperature (see Definition ~\ref{def:thex}). These metrics provide a consistent framework for assessing our protocol throughout the remainder of this work.


\begin{definition}\label{def:thex}
    Let $\mathcal{U}_1^G$ denote the set of single-mode Gaussian unitary operators. The thermal excitation of a single-mode state $\rho_1$ is
    \begin{equation}\label{eq:th_ex}
\ThEx(\rho_1):=\min_{W\in\mathcal{U}^G_1}\tr(\hat{a}^\dagger \hat{a} W\rho_1 W^\dagger).
    \end{equation}
    From the Bose-Einstein distribution, the corresponding effective inverse temperature is given by
    \begin{equation}
        \beta_{eff}(\rho_1):=\frac{1}{\omega}\ln\frac{\ThEx(\rho_1)+1}{\ThEx(\rho_1)}.
    \end{equation}
\end{definition}
The purpose of minimization in Eq.~\eqref{eq:th_ex} is to exclude the contribution of displacement and squeezing to the excitation number, such that the mean excitation number is caused solely by thermal noises.
For a single-mode Gaussian state $\rho^G$, $\ThEx(\rho^G)$ accounts for the thermal excitation contribution to the mean excitation number $\nb(\rho^G)=\tr(\hat{a}^\dagger\hat{a}\rho^G)$, aside from the contributions from squeezing and displacement. Moreover, $\rho^G$ can be transformed by a Gaussian unitary to the Gibbs state $\tilde \tau(\rho^G)$ with inverse temperature $\beta_{eff}(\rho^G)$, whose mean photon number is $\ThEx(\rho^G)$ (see for example, Ref.~\cite{weedbrook2012gaussian}, or Appendix~\ref{app:A1} for a brief proof). Notably, because $\tilde \tau(\rho^G)$ is related to $\rho^G$ by a Gaussian unitary, the von Neumann entropy $S(\rho^G):=-\tr(\rho^G\ln\rho^G)$ and purity $P(\rho^G):=\tr[(\rho^{G})^2]$ are also monotonic on $\ThEx(\rho^G)$. Thus, the thermal excitation $\ThEx(\rho^G)$ can also be employed as an effective measure of nonequilibrium or purity of the Gaussian state $\rho^G$.
Every non-Gaussian state $\rho^{NG}$ has a Gaussian counterpart $\tilde\rho^G(\rho^{NG})$, such that the first and second moments of $\rho^{NG}$ and $\tilde\rho^G(\rho^{NG})$ coincide. Thus, it is natural to employ $\ThEx[\tilde\rho^G(\rho^{NG})]$ to measure the thermal excitation in $\rho^{NG}$, which equals to $\ThEx(\rho^{NG})$ by Definition~\ref{def:thex} \cite{nimmrichter2017quantum}.

\section{Cooling with Gaussian operation \label{sec:CoolingGO}}

\subsection{Limitation of Gaussian cooling \label{sec:LimGaussianCooling}}
With the preliminaries established in Sec.~\ref{sec:Preliminary}, we now analyze the constraints and limitations inherent to Gaussian cooling protocols. 
The unitary operations in both recharging and thermalizing routines are restricted to be Gaussian in this section, so we first consider the ability of a multimode Gaussian unitary in minimizing the thermal excitation in the  first mode, and prove the following lemma (the proof is in Appendix~\ref{app:A2}).
\begin{lemma}\label{le:global_U}
    Consider a $J$-mode system initially in a Gaussian product state $\rho_J=\otimes_{j=1}^J\rho_j$. 
    After the action of a global Gaussian unitary, the minimum of the thermal excitation of Mode 1 in the output is
    \begin{equation}\label{eq:le1}
        \min_{W\in\mathcal{U}^G_J}\ThEx\left(\tr_{\backslash1}(W\rho_J W^\dagger)\right)
        =\min_j \ThEx(\rho_j),
    \end{equation}
    where $\mathcal{U}^G_J$ is the set of $J$-mode Gaussian unitary operators.
\end{lemma}

Applying Lemma~\ref{le:global_U} to the scheme of HBAC, we obtain Theorem~\ref{thm:1}, which specifies the necessary and sufficient condition on sub-bath cooling via Gaussian operations.
\begin{theorem}\label{thm:1}
By Gaussian HBAC protocols, the effective inverse temperature of $S$ can reach $\beta^*=\lambda\beta$ $(\lambda>1)$ if and only if $\omega_N\geq\lambda\omega_0$.
\begin{proof}
In our setting, the unitary $U$ in the thermalizing routine is both Gaussian and energy-preserving. Thus, the reservoir modes that can exchange energy with the machine should be resonant with the machine modes.
It means that the frequencies of the effective reservoir modes are no larger than $\omega_N$.

In our generic framework of HBAC, the output state $\rho^{(l+1)}_{SM}$ of $SM$ after the $l$th iteration is related to $\rho^{(l)}_{SM}$ as $\rho^{(l+1)}_{SM}=\tr_{R_{l}}[W_l(\rho^{(l)}_{SM}\otimes\tau_{R_{l}})W_l^\dagger]$, where $W_l=U_lV$, $V$ acts on $SM$, $U_l$ acts on $SMR_l$, and $\rho^{(1)}_{SM}=\tau_S\otimes\tau_M$ (See Fig.~\ref{fig:Overview} (a)). It follows that after $L$ iterations,
\begin{equation}
    \rho_S^{(L)}=\tr_{MR}[W(\tau_S\otimes\tau_M\otimes\tau_R)W^\dagger],
\end{equation}
where $\tau_R=\otimes_{l=1}^{L}\tau_{R_l}$ is a multimode reservoir composed of different Gibbs state denoted by $\tau_{R_l}$, and $W=W_L\cdots W_1$ is a Gaussian unitary. By Lemma~\ref{le:global_U}, the mean excitation number in the system mode is bounded as $\ThEx(\rho^{(L)}_S)\geq\min_j\ThEx(\tau(\omega_j))=\ThEx(\tau(\omega_N))$. It follows that the effective temperature of the system mode satisfies
    \begin{equation}
        \beta^*=\beta_{eff}(\rho^{(L)}_S)\leq\frac{\omega_N}{\omega_0}\beta.\label{eq:cooling_limit}
    \end{equation}
Therefore, for the Gaussian HBAC protocol to reach $\beta^*=\lambda\beta$, it is necessary that $\omega_N\geq\lambda\omega_0$. This proves the ``only if'' direction. Conversely, if $\omega_N\geq\lambda\omega_0$, then $\beta^*$ can be reached by setting the recharging routine as $V=\exp[-i\chi(\hat{a}^\dagger\hat{b}_N+\hat{a}\hat{b}_N^\dagger)t]$ for some proper $t$.
\end{proof}
\end{theorem}





This theorem directly implies the necessity of a machine for sub-bath cooling under Gaussian operations, which is consistent with the no-go theorem \cite{PhysRevLett.124.010602} stating that cooling below the bath temperature via Gaussian thermal operations is impossible. Our results unveil further restrictions of Gaussian operations in cooling protocols.

In algorithmic cooling (also known as nonequilibrium concentration in the context of resource theories), one starts from $J$ copies of state $\rho$, applies a joint unitary $U$ on them, and aims at purifying the output state $\rho'=\tr_{\backslash 1}(U\rho^{\otimes J} U^\dagger)$ of the first system, such that $S(\rho')<S(\rho)$ or $P(\rho')>P(\rho)$. This task is achievable for $J\geq3$ if $\rho$ and $\rho'$ are qubit states~\cite{PhysRevA.110.022215}, and for $J\geq2$ if $\rho$ and $\rho'$ are $d$-dimensional ($d\geq3$) systems \cite{hsieh2025informational}. 
In contrast, Lemma~\ref{le:global_U} indicates that it is impossible to do algorithmic cooling using only Gaussian unitary operations. Precisely, the rhs of Eq.~\eqref{eq:le1} equals $\ThEx(\rho)$ because $\rho_j=\rho,\forall j$. Thus, the thermal excitation of the first mode in the output cannot be lower than the input. Consequently, it is impossible to achieve algorithmic cooling using only Gaussian unitary operations. 


As discussed in Sec.~\ref{sec:Preliminary}, HBAC with incoherent control permits only unitary operations that conserve the total energy. Consequently, the recharging unitary $V$ is constrained to be energy-preserving. In this setting, entropy removal requires two machines with distinct roles. The colder machine extracts entropy from the system, while the hotter machine absorbs the excess energy and entropy generated during cooling. Consequently, the cooling protocol involves energy-preserving interactions with two machines —one at an inverse temperature $\beta$ and the other at a higher temperature $\beta_H < \beta$. Because Gaussian interactions are generated by quadratic Hamiltonians, they are energy-preserving only when the system mode is resonant with the machine modes. Under this resonance condition, Lemma~\ref{le:global_U} implies that cooling to an inverse temperature greater than $\beta$ is impossible. This result stands in sharp contrast to the discrete-variable case \cite{PhysRevLett.123.170605}, where incoherent-control HBAC protocols can reach the same cooling limit as those based on coherent control.



In the PPA-like protocols, the machine consists of two parts: the reset part, which is refreshed to the Gibbs state in the thermalizing routine in each round, and the scratch part, which remains in the same state and thus carries forward memory. For both the qubits case (where the system is a single qubit and the machine is composed of qubits)~\cite{PhysRevLett.116.170501_PPA} and the arbitrary finite-dimensional case~\cite{PhysRevApplied.14.054005}, it is shown that the memory effect of the scratch can provide an exponential advantage over the memoryless scenario. On the contrary, Theorem~\ref{thm:1} implies that the memory effect does not improve the cooling limit in the Gaussian setting. 
From Theorem~\ref{thm:1}, when $\omega_N=\lambda \omega_0$, the inverse temperature of the system can at most reach $\beta^*=\lambda\beta$ by Gaussian HBAC protocols, which include the PPA-like protocol as a special case. 
On the other hand, this limit can be achieved in one round by simply swapping the states of $S$ and $M_N$ in the recharging routine, which is memoryless. 
Hence, our result shows that, in order to exploit the memory effect, non-Gaussianity in the operations is necessary. 

\subsection{Energy efficiency of Gaussian cooling \label{sec:EnergyEfficiencyGaussian}}

As shown in the last subsection, the cooling limit can be achieved within one round. During the recharging routine, the system and the machine interact via a Gaussian unitary $V$, and the joint state becomes $\rho_{SM}'=V(\tau_S\otimes\tau_M)V^\dagger$.
During the thermalizing routine, 
the machine is reset to the Gibbs state for later use, 
and hence an amount of heat $Q=\tr[H_M(\rho_M'-\tau_M)]$ is dissipated to the reservoir. Using the formalism in Ref.~\cite{reeb2014improved}, the dissipated heat is lower bounded as 
\begin{equation}\label{eq:landauer1}
    \beta Q\geq\tilde{\Delta}S_S,
\end{equation}
where $\tilde{\Delta}S_S=S(\tau_S)-S(\rho'_S)$ is the decrement of entropy in the target system. 
For example, when one bit of information is erased, a qubit state is transformed from $\tau_S=\frac{\iden}{2}$ to a pure state $\rho'_S=|0\rangle\langle 0|$, leading to a decrease in the qubit’s entropy of $\tilde{\Delta}S_S=\ln2$. Eq.~\eqref{eq:landauer1} then implies that the dissipated heat $Q$ is lower bounded by $k_BT\ln2$, thereby recovering the well-known Landauer's bound~\cite{landauer1961irreversibility}. Thus, Eq.~\eqref{eq:landauer1} can be viewed as a generalized form of Landauer’s bound, extending it to scenarios with arbitrary initial and final states.
Further, the difference is usually called the entropy production and can be calculated as~\cite{reeb2014improved}
\begin{equation}
    \Sigma_N(V,H_M):=\beta Q-\tilde{\Delta}S_S=D[\rho'_M||\tau_M]+I_{S:M}(\rho'_{SM}),
\end{equation}
where $I_{S:M}$ is the mutual information and $D[\cdot \| \cdot]$ is the relative entropy. A lower value of $\Sigma_N(V,H_M)$ indicates a higher efficiency of the protocol. The efficiency is maximized when $\Sigma_N(V,H_M)=0$. In this limit, all redundant heat vanishes, and the dissipated heat 
reaches its lower bound, i.e., $\beta Q = \tilde{\Delta}S_S$. In the following, we derive the optimal Gaussian unitary $V$ in the recharging routine and the optimal spectrum $\omega_1,\dots, \omega_{N-1}$ of the machine, such that $\Sigma_N(V,H_M)$ reaches the minimum for given $N$.
For an arbitrary choice of the spectrum of the machine, the following theorem identifies successive bimode swaps as the most efficient unitary in the recharging routine. The proof is given in Appendix~\ref{app:A3}.
\begin{theorem}\label{thm:2}
    For any choice of $\omega_1,\dots,\omega_{N-1}$, among all the protocols that can achieve the cooling limit $\beta^*=\frac{\omega_N}{\omega_0}\beta$, the most efficient unitary in the recharging routine is
    \begin{eqnarray}
        V^*& := & \mathrm{arg}\min_{V\in\mathcal{U}^G_{N+1}}\Sigma_N(V,H_M)\nonumber\\
        & = & \Swap(S,M_N)\circ\dots\circ\Swap(S,M_{j_0}),
    \end{eqnarray}
    where $\Swap(\cdot,\cdot)$ is the unitary which swaps the states of two modes, $j_0$ is the minimum index such that $\omega_{j_0}>\omega_0$. Besides, the minimum entropy production reads
    \begin{equation}
        \Sigma_N^*(H_M)=D[\tau(\omega_0)\| \tau(\omega_{j_0})]+\sum_{j=j_0+1}^N D[\tau(\omega_{j-1}) \| \tau(\omega_j)],\label{eq:ent_pro_star}
    \end{equation}
    where $\tau(\omega_j)=(1-e^{-\beta\omega_j})\sum_{n=0}^\infty e^{-n\beta\omega_j}|n\rangle\langle n|$ is the Gibbs state of a harmonic oscillator with frequency $\omega_j$ and inverse temperature $\beta$.
\end{theorem}
Prior to this result for general Gaussian cooling protocols, to the best of our knowledge, the most efficient recharging unitary had been tackled analytically only for the collisional cooling for qubits \cite{PhysRevLett.134.070401} where the system is a single qubit, the machine consists of $N$ qubits, and the system qubit interacts successively with machine qubits, one at a time and each qubit only once. 

Clearly, the machine modes with $\omega_j\leq\omega_0$ do not play any role in improving the cooling limit or energy efficiency. Therefore, for the purpose of optimizing $\Sigma_N^*(H_M)$ for a given $N$, one should set $\omega_1>\omega_0$. With this, Eq.~(\ref{eq:ent_pro_star}) can be rewritten as
\begin{eqnarray}
    & &\Sigma_N^*(H_M) = \sum_{j=1}^N D[\tau(\omega_{j-1})\| \tau(\omega_j)]\nonumber\\
    & = & \sum_{j=1}^N (\bar{n}_{j-1}+1)\ln\frac{\bar{n}_j+1}{\bar{n}_{j-1}+1}+\bar{n}_{j-1}\ln\frac{\bar{n}_{j-1}}{\bar{n}_j},
\end{eqnarray}
where $\bar{n}_j:=\bar{n}(\tau(\omega_j))=\frac{1}{e^{\beta\omega_j}-1}$.
Next, we focus on optimizing the entropy production over the machine spectrum $\omega_1,\dots,\omega_{N-1},$
\begin{equation}
    \Sigma_N^{**}=\min_{\omega_1,\dots,\omega_{N-1}}\Sigma_N^*(H_M).
\end{equation}
We verify that $\frac{\partial^2\Sigma_N^*(H_S)}{\partial \bar{n}_j^2}=\frac{\nb_{j-1}}{\nb_j^2}-\frac{\nb_{j-1}+1}{(\nb_j+1)^2}+\frac{1}{\nb_j(\nb_j+1)}>0$ for $\bar{n}_{j-1}\geq \bar{n}_j$, and hence, $\Sigma_N^*(H_S)$ is a strictly convex function of $\bar{n}_1,\dots,\bar{n}_{N-1}$ for $\bar{n}_0\geq \bar{n}_1\geq \dots\geq \bar{n}_{N-1}\geq \bar{n}_N$. Therefore, the minimum of $\Sigma_N^*(H_S)$ is reached if and only if $\frac{\partial \Sigma_N^*(H_S)}{\partial \bar{n}_j}=0$, which gives
\begin{equation}
    g_{j+1}-g_{j} = (e^{g_j-g_{j-1}}-1)\cdot\frac{1-e^{-g_{j}}}{1-e^{-g_{j-1}}},\ j=1,\dots,N-1,\label{eq:opt_spec}
\end{equation}
where we define $\beta\omega_j=g_j (j=0,\dots,N)$ for convenience, and $g_0$ and $g_N$ are the boundary conditions. Eq.~(\ref{eq:opt_spec}) can be solved numerically and plotted in Fig.~\ref{fig:TL}, see also Fig.~\ref{fig:Traj} in Appendix \ref{app:A3}. 
For a finite number of machine modes $N$, we observe a substantial reduction in heat dissipation as $N$ increases.
In particular, the dissipation decreases by more than half with the introduction of just one additional machine mode ($N=2$).

\begin{figure}[htbp]
    \centering
    \includegraphics[width=0.42\textwidth]{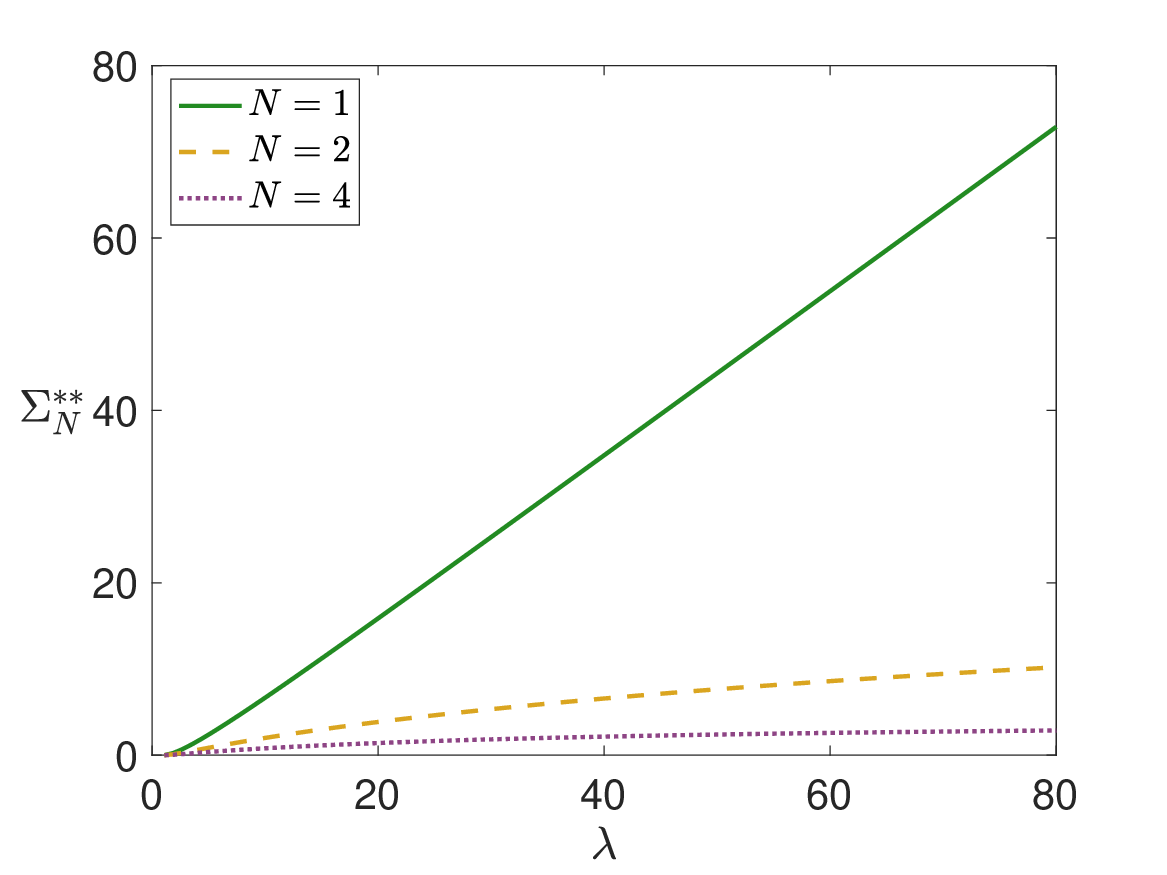}
    \caption{Plot of minimum heat dissipation $\Sigma_N^{**}$ as the function of as a function of the cooling ratio $\lambda = \beta^*/\beta$.  Here, $\lambda$ defines the target inverse temperature relative to the initial one. The initial mean excitation number is fixed at $n_0 = 10$, corresponding to the normalized effective inverse temperature $\beta \omega_0 = \ln(11/10)\approx 0.0953$. For each $\lambda$, $\Sigma_N^{**}$ is optimized over the energy gap of $N = {1, 2, 4}$ machine modes. We observe that heat dissipation increases as the cooling demand increases (larger $\lambda$) for a fixed number of machine modes. However, introducing additional machine modes significantly reduces this dissipation. }
    \label{fig:TL}
\end{figure}

%
In the limit where $N$ is large, such that $(g_j-g_{j-1})^k N^2$ is infinitesimal for $k\geq3$, Eq.~(\ref{eq:opt_spec}) can be solved analytically as
\begin{equation}
    g^*_j=2\ln\left\{-\coth\left[\frac{j}{2N}\ln\tanh\frac{g_N}{4}+\frac{N-j}{2N}\ln\tanh\frac{g_0}{4}\right]\right\}. \label{eq:analyTraj}
\end{equation}
Therefore, the optimal entropy production for large $N$ reads
\begin{equation}
    \Sigma^{**}_{N\gg1}=\frac{1}{2N}\left[\frac{1}{2}\ln\frac{\tanh(\beta\omega_N/4)}{\tanh(\beta\omega_0/4)}\right]^2+o(\frac{1}{N^2}).\label{eq:Sigma_N_large}
\end{equation}
The details of the calculation can be found in Appendix~\ref{app:A3}.

Eq.~\eqref{eq:Sigma_N_large} shows that the entropy production $\Sigma_N^{**}$ converges to zero in the limit $N\rightarrow\infty$, demonstrating that the generalized Landauer's bound in Eq.~\eqref{eq:landauer1} can be approached when the time cost diverges, which is consistent with the result in Ref.~\cite{PRXQuantum.4.010332}. Besides, Eq.~(\ref{eq:Sigma_N_large}) scales as $1/N$. 
Building on Ref.~\cite{PhysRevLett.134.070401}, which establishes $1/N$ scaling for optimal entropy production in collisional models for finite-dimensional systems, we extend this finding to the continuous-variable regime. Furthermore, since our optimization is over both the genuine multipartite Gaussian unitary operators and the spectrum of the machine, Eq.~(\ref{eq:Sigma_N_large}) actually gives the general and reachable lower bound of dissipated energy in Gaussian cooling for given $\omega_0$, $\omega_N$, and $N\gg 1$. Our result shows that the collective advantage \cite{PhysRevLett.131.210401} disappears when the Gaussian condition is imposed, complementing the result for thermodynamics length \cite{Abiuso_2022}, where the continuous time evolution of a system is studied. Our results show that Gaussian HBAC cannot surpass the $1/N$ convergence rate toward the generalized Landauer limit, thus offering no collective advantage. This limitation arises because Gaussian operations lack the strong nonlocal, long-range interactions required for collective enhancement~\cite{PhysRevLett.131.210401}. This observation also complements the findings on thermodynamic length in classical Gaussian dynamics~\cite{Abiuso_2022}.   
 
\section{Cooling with Non-Gaussian operation \label{sec:CoolingNGU}}
\subsection{Single-shot cooling \label{sec: SingleShotNonGaussian}}

In this section, we investigate the cooling enhancement offered by $p$-excitation exchange operation induced by the interacting Hamiltonian $H_I^{(p)}(\chi) = \chi (\hat{a} \hat{b}^{\dag p} + \hat{a}^\dag \hat{b}^p)$ with $p \in \{1,2,3,...\}$. The case $p=1$ corresponds to a beam-splitter-type interaction which yields the optimal recharging routine among all Gaussian operations. For $p \geq 2$, the dynamics become non-Gaussian and are realizable through multimode mixing in experiments \cite{ding2017quantum,chang2020observation}. 
%
In particular, the quantum absorption refrigeration \cite{nimmrichter2017quantum,correa2014quantum}, which relies on three-wave mixing,
has been demonstrated in trapped-ion systems \cite{maslennikov2019quantum}. A recent implementation of autonomous quantum refrigerator is also reported in superconducting qubit \cite{aamir2025thermally}.
%

Specifically, we consider a single-mode system and machine, each initialized in a Gibbs state ($\tau_S$ and $\tau_M$) at inverse temperature $\beta$. Their bare Hamiltonian is  $H_0 = \omega_0 \hat{a}^\dag \hat{a} + \omega_1 \hat{b}^\dag \hat{b}$. The system and the machine then interact through the unitary $V^{(p)}(t) = \exp(-i H^{(p)} t)$ generated by $H^{(p)} = H_0 + H_I^{(p)}(\chi)$. Obtaining a closed analytical expression for the exact total excitation under this nonlinear interaction is extremely challenging. As a result, earlier works have either focused on short-time dynamics with coherent or Fock initial states \cite{birrittella2020phase}, or relied on effective long-time equilibration toward a steady state. In our work, we follow the short-time approach but consider the more realistic case in which the initial states are Gibbs states. Now, the single-shot cooling operation is defined as 
\begin{equation}
    \mathcal{E}_t(\tau_S) = \tr_M\left[ V^{(p)}(t) \left( \tau_S \otimes \tau_M \right) V^{(p)\dag}(t)  \right ]. 
\end{equation}
In the Heisenberg picture, the mean excitation number of the system evolves as
\begin{align}
    \bar{n}_S(t) &= \tr\left[ \tau_S  \mathcal{E}_t^\dag (\hat{n}_S ) \right ]
\end{align}
where $\hat{n}_S = \hat{a}^\dag \hat{a}$ and $\mathcal{E}_t^\dag(\hat{n}_S) = \tr_M \left[ (\iden \otimes \sqrt{\tau_M}) V^{(p)\dag}(t) (\hat{n}_S \otimes \iden ) V^{(p)}(t) (\iden \otimes \sqrt{\tau_M}) \right]$. In the $\chi t\ll 1$ limit, the mean excitation number is approximated as  (see Appendix ~\ref{app:B1})
\begin{align}
    \bar{n}_{S}(t) 
    &\approx \bar{n}_S - \chi^2 t^2 \left[ (1+\bar{n}_M)^p \bar{n}_S  - \bar{n}_M^p (1+\bar{n}_S )\right] p!,  \label{eq:nst} 
\end{align}
where $\bar{n}_{S}:=\bar{n}_{S}(0)$ and  $\bar{n}_{M}:=\bar{n}_{M}(0)$ are the initial mean excitation numbers of the system and machine, respectively. After a $p$-excitation exchange, the transient state of the system is no longer in the form of a Gibbs state. Nevertheless, as we show in Appendix~\ref{app:B1}, the thermal excitation of $\rho'_S(t)=\mathcal{E}_t(\tau_S)$ equals to its mean excitation number, i.e., $\ThEx(\rho'_S(t))=\bar{n}_S(t)$. Henceforth, we employ $\bar{n}_S(t)$ to quantify the cooling performance.
If $\bar{n}_{S}(t) < \bar{n}_{S} $ for $\chi t\ll1$, we say that the system is cooled beyond its initial temperature, i.e., $\beta_{eff}(\rho'_S(t)) > \beta$. The necessary and sufficient condition for sub-bath cooling with the $p$-excitation exchange Hamiltonian is given by the following theorem.

\begin{theorem}\label{thm:NonGCooling}
Under a single-shot $p$-excitation exchange cooling operation $\mathcal{E}_t$ with $\chi t\ll1$, the necessary and sufficient condition for sub-bath cooling is $p \omega_1 >  \omega_0$.
\end{theorem}
\begin{proof}
From Eq.~\eqref{eq:nst}, the sub-bath cooling $\bar{n}_{S}(t) < \bar{n}_{S} $ is equivalent to  $\chi^2 t^2 \left[ (1+\bar{n}_M)^p \bar{n}_S  - \bar{n}_M^p (1+\bar{n}_S )\right] p! >0$, which gives 
    \begin{equation}
        \bar{n}_{S} > 
        \frac{\bar{n}_M^p}{(1+\bar{n}_M)^p-\bar{n}_M^p}.\label{eq:n_min_k}
    \end{equation}
Because $(1+\bar{n}_M)/\bar{n}_M=e^{\beta \omega_1}$, the right-hand side is equal to $\frac{1}{e^{\beta p\omega_1}-1}$. From $\bar{n}_{S}=\frac{1}{e^{\beta \omega_0}-1}$, Eq.~\eqref{eq:n_min_k} holds if and only if $p \omega_1 >  \omega_0$. 
\end{proof}
\begin{figure*}
    \centering
    \includegraphics[width=0.32\textwidth]{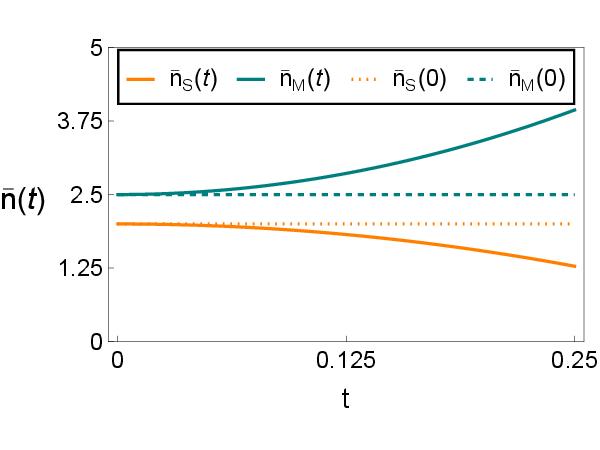}
    \includegraphics[width=0.32\textwidth]{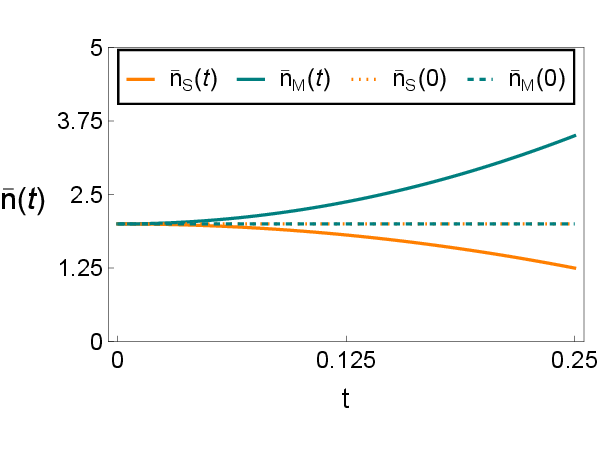}
    \includegraphics[width=0.32\textwidth]{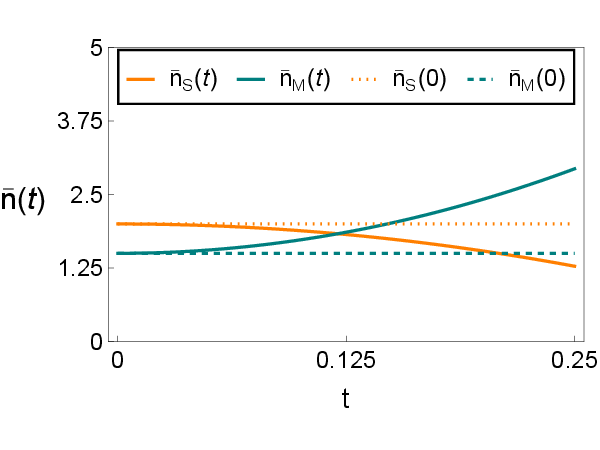}
    \caption{ Plot of ${\bar{n}_S(t), \bar{n}_M(t)}$  with initial condition (from left to right)  $\bar{n}_S < \bar{n}_{M}$,  $\bar{n}_S = \bar{n}_{M}$,  $\bar{n}_S > \bar{n}_{M}$. The orange (teal) solid line indicates $\bar{n}_S(t)$ ($\bar{n}_{M}(t)$). The dashed line indicates the initial mean excitation number where $\bar{n}_S=2$, $\bar{n}_M={1.5,2.0,2.5}$.  }
    \label{fig:NGUexamplek2}
\end{figure*}
%
We find that requirement of the machine frequency in Theorem~\ref{thm:NonGCooling}, $\omega_1 > \omega_0/p$ is less restrictive than the one in  Theorem~\ref{thm:1}, $\omega_1 > \omega_0$.
This implies that cooling with a machine whose frequency is equal to or lower than the system frequency is possible for $V^{(p)}(t)$ with $p\geq 2$. In Fig.~\ref{fig:NGUexamplek2}, we give an example of cooling with $V^{(2)}(t)$ for three initial conditions $\bar{n}_S < \bar{n}_{M}$,  $\bar{n}_S = \bar{n}_{M}$,  $\bar{n}_S > \bar{n}_{M}$. We show that sub-bath cooling is possible even in the cases of $\bar{n}_S < \bar{n}_{M}$ ($\omega_1 < \omega_0$),  $\bar{n}_S = \bar{n}_{M}$ ($\omega_1 = \omega_0$), where Gaussian cooling fails. For the case of $\bar{n}_S > \bar{n}_{M}$ ($\omega_1 > \omega_0$), the system's mean excitation drops beyond the machine's mean excitation for $\chi t\ll 1$ and 
    \begin{equation}
    t > t_c = \frac{1}{\chi} \sqrt{\frac{\bar{n}_S - \bar{n}_{M}}{ \left[ (1+\bar{n}_M)^p \bar{n}_S  - \bar{n}_M^p (1+\bar{n}_S )\right]p! }}.
\end{equation}  

The lower bound of the cooling condition in Theorem~\ref{thm:NonGCooling} corresponds to the resonance condition $\omega_0 = p \omega_1 $, under which the $p$-excitation exchange interaction $V^{(p)}$ becomes energy preserving. In this resonant regime, no cooling enhancement is possible. Thus, Theorem~\ref{thm:NonGCooling} generalizes the previously known no-cooling condition for cooling through trilinear Hamiltonian~\cite{nimmrichter2017quantum} to arbitrary $p$-excitation exchange processes. 

\subsection{Iterative cooling \label{sec:NonGaussianIterative}}
In the previous subsection, we showed that $V^{(p)}(t)$ provides a cooling advantage in the short-time regime. We now investigate the cooling behavior when this operation is applied repeatedly. Within the HBAC framework, we iterate the recharging step $V^{(p)}(t)$, followed by full thermalization of the machine for $L$ rounds (see Fig. \ref{fig:Overview} (b)). In each round, the system interacts with a refreshed machine $\tau_M$ with a fixed machine gap for the same short duration $\chi t\ll 1$. Since $\mathcal{E}_t$ is a dissipative channel, the system state converges to a steady state after sufficiently many iterations. To determine whether this steady state is thermal, we evaluate the thermal Fano factor \cite{nimmrichter2017quantum}. We find that it approaches zero in the large-$L$ limit, indicating that the steady state converges to a Gibbs state (see Appendix~\ref{app:B3}). Thus, in the asymptotic limit, the repeated action of $\mE_t$ maps any initial Gibbs state to another Gibbs state, while allowing for nonequilibrium dynamics during intermediate iterations. Since the steady state in the large-$L$ limit is thermal, its temperature can be characterized entirely by its mean excitation. Using this fact, we derive the HBAC cooling limit for the $p$-excitation exchange operation.

\begin{theorem}\label{thm:NonGCooling2}
     In the large $L$ limit, the Markovian collisional model governed by the unitary interaction $V^{(p)}(t)$ maps a Gibbs state to another Gibbs state. The asymptotic effective inverse temperature can reach $\beta^* = \frac{\omega_1}{\omega_0} p \beta$.
\end{theorem}

\begin{proof}
The mean excitation number after $L$ iterations for $p$-excitation exchange is given by
\begin{align}
    \bar{n}_S^{(L)} 
    &= \tr\Big[ \tau_S  \underbrace{\mathcal{E}_t^\dag \circ ... \circ \mathcal{E}_t^\dag }_{L\  \text{times} } (\hat{n}_S ) \Big ].
\end{align}
From this, we derive the closed-form expression of the mean excitation number after $L$ iterations of cooling as (see Appendix~\ref{app:B2})
\begin{equation}
    \bar{n}_S^{(L)} \approx \bar{n}_S^{(0)}(1-a)^L + b \frac{1-(1-a)^L}{a}
\end{equation}
where $a = \chi^2 t^2[(1+\bar{n}_M)^p + \bar{n}_M^p]p!$ and $b = \chi^2 t^2 \bar{n}_M^p p!$. 
In the large-$L$ limit, the mean excitation approaches
\begin{equation}
    \bar{n}_S^{(\infty)}:=\lim_{L\rightarrow\infty} \bar{n}_S^{(L)} =\frac{1}{e^{p\beta \omega_1}-1}, \label{eq:emanPhotonLB}
\end{equation}
which saturates to the minimum allowed mean excitation of the system for sub-bath cooling. Hence, it cannot be reduced even further by $p$-excitation interaction from Theorem~\ref{thm:NonGCooling}.
Besides, because the asymptotic state of the system is a Gibbs state, we have $\bar{n}_S^{(\infty)}=\frac{1}{e^{\beta^* \omega_0 }-1}$.
Equation~\eqref{eq:emanPhotonLB} then implies $\beta^* = \frac{\omega_1}{\omega_0} p \beta $.
\end{proof}

Theorem~\ref{thm:NonGCooling2} highlights that the chain of  non-Gaussian collisions $\mathcal{E}_t$ achieves the cooling limit equivalent to a single excitation Gaussian exchange with a machine possessing a $p$-fold increased energy gap (see Fig. \ref{fig:Overview} (d)). With increasing nonlinearity $p$, not only does the asymptotic mean excitation decreases exponentially, but the convergence rate toward the asymptotic mean excitation is also increased by $(p!)$-fold (see Fig. \ref{fig:IterativeCooling}).
%
For $p = 1$, we recover $\bar{n}_S^{(\infty)} = \bar{n}_M$. Since only single-excitation exchanges are permitted, the energy exchange per interaction is limited. This agrees with our previous result that the most optimal Gaussian recharging operation fails to surpass the limit $\beta^* = \frac{\omega_1}{\omega_0} \beta $.  For $p \geq 2$; however, the $p$-excitation exchange Hamiltonian allows multiple quanta to be transferred in a single interaction. This accelerates the cooling process, enabling the system temperature to converge more rapidly and decrease by a factor of $p$. Since higher $p$-nonlinearity implies stronger non-Gaussianity, the observed $p$-fold cooling enhancement emphasizes that non-Gaussianity is a key resource for surpassing the cooling limits imposed by Gaussian dynamics.

\begin{figure}
    \centering
    \includegraphics[width=1\linewidth]{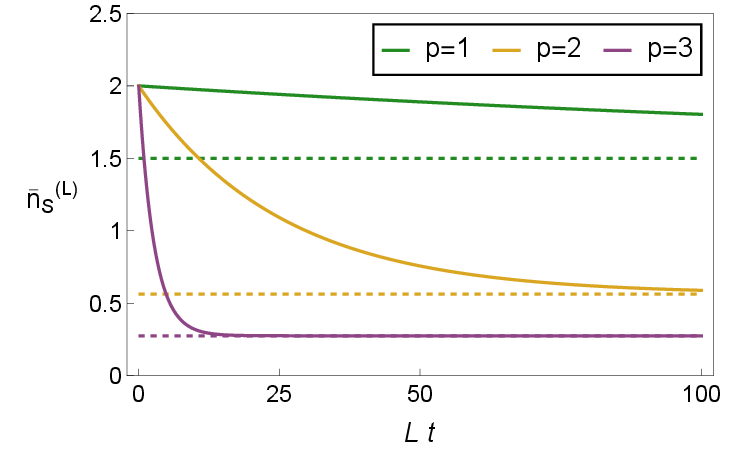}
    \caption{The time evolution of $\bar{n}_S^{(L)}:=\bar{n}(L t )$ under iterative cooling for $L=1,2,3,...,2\times 10^4$ and $t=5\times10^{-3} s$. Within each iteration, the collisional channel is governed by the $p$-excitation exchange operation. Initially, $\bar{n}_S^{(0)} = 2$, and $\bar{n}_M^{(0)}=1.5$. In the large $L$ limit, the minimum achievable mean excitation saturates to  Eq.~(\ref{eq:emanPhotonLB}), indicated by the dashed line. As $p$-nonlinearity increases, the minimum mean excitation decreases. For $p\geq 2$, the asymptotic mean excitation decreases below $\bar{n}_M^{(0)}$. }
    \label{fig:IterativeCooling}
\end{figure}

\section{Conclusion}
%
%
In this paper, we have investigated the fundamental cooling limits of multimode Gaussian operations and the cooling advantage offered by a specific non-Gaussian operation, the $p$-excitation exchange operation. Optimizing over all Gaussian protocols, we prove that the most optimal cooling strategy is successive $\Swap$ operations with machine modes of increasingly larger energy gaps.  Since $\Swap$ interactions exchange only a single excitation per round, Gaussian protocols are intrinsically limited in the control complexity required for more energy extraction during the recharging routine. Consequently, Gaussian operations cannot cool a target system when the available machine modes have gaps that are equal to or smaller than the system's. This limitation rules out nonequilibrium concentration, excludes any cooling advantage from incoherent Gaussian control, and shows that Gaussian protocols gain no benefit from memory effects. Beyond single-excitation exchange, we demonstrate that the $p$-excitation exchange yields a p-fold cooling enhancement. This improvement arises from genuine non-Gaussianity, which leverages increased control complexity to extract more energy. The experimental realization of higher-order $p$-excitation is challenging, but has been demonstrated in superconducting circuits \cite{aamir2025thermally, mundhada2019experimental} and trapped ions \cite{maslennikov2019quantum}.

%

In the future work, the cooling efficiency of the $p$-excitation exchange operation and the corresponding collective advantage remain to be established. The potential of memory effects induced by the $p$-excitation exchange to further enhance cooling will be also addressed in the future. 



\begin{acknowledgments}
The authors thank Paolo Abiuso, Pharnam Bakhshinezhad, and Paul Skrzypczyk for helpful discussions. This project is supported by the National
Research Foundation, Singapore through the National
Quantum Office, hosted in A*STAR, under its Centre for
Quantum Technologies Funding Initiative (S24Q2d0009);
by the Ministry of Education, Singapore, under the Tier
2 grant “Bayesian approach to irreversibility” (Grant
No. MOE-T2EP50123-0002); by NSFC under Grant No.
11774205, and the China Scholarship Council (award to
Xueyuan Hu for 1 year’s study abroad at the National
University of Singapore).
\end{acknowledgments}

\bibliography{apssamp}
\pagebreak
\onecolumngrid
\appendix
\section{Appendices for Gaussian cooling}
\subsection{Preliminaries: Gaussian states and Gaussian unitary operations}\label{app:A1}
Consider a $J$-mode bosonic system, whose Hamiltonian is written as $H=\sum_{j=1}^J\omega_j\hat{a}_j^\dagger \hat{a}_j$, where $\hat{a}_j$, $\hat{a}_j^\dagger$, and $\omega_j$ are the annihilation and creation operators, and the frequency of the $j$th mode.
A Gaussian state of this system is fully characterized by its first and second moments defined as $\FirM_{j} = \FirM_{j+J}^{*} = \langle \hat{a}_{j} \rangle$ for $j = 1,\cdots,J$ and $\SecM_{j,k} = \SecM_{k+J,j+J}= \frac{1}{2}\langle \{\hat{a}_{j},\hat{a}_{k}^{\dagger}\} \rangle - \langle \hat{a}_{j}\rangle \langle \hat{a}_{k}^{\dagger}\rangle $, $\SecM_{j,k+J} = \SecM_{j+J,k}^{*} = \frac{1}{2}\langle \{\hat{a}_{j},\hat{a}_{k}\} \rangle - \langle \hat{a}_{j}\rangle \langle \hat{a}_{k}\rangle $ for $j,k = 1,\cdots,J$, respectively.
Accordingly, the first moment vector $\FirM$ and the second moment matrix $\SecM$ can be written as 
\begin{eqnarray}
    \FirM=\begin{pmatrix}
        \vec\alpha^* \\ \vec\alpha
    \end{pmatrix},\ \SecM = \begin{pmatrix}
        \mu^* & \nu \\ \nu^* & \mu
    \end{pmatrix},
\end{eqnarray}
where $\vec\alpha$ is a $J$-dimensional vector and $\mu$ and $\nu$ are $J\times J$-dimensional matrices that are Hermitian and symmetric, respectively. 
Moreover, the second moment must obey the uncertainty principle, which corresponds to $\SecM+\mathbf{Z}\geq \mathbf{0}$ with a $2J\times 2J$ matrix $\mathbf{Z}:= \frac{1}{2}\begin{pmatrix}
	\mathbb{I} & \mathbf{0}\\ \mathbf{0} & -\mathbb{I}
\end{pmatrix}$, where $\mathbb{I}$ is the identity matrix. With this, the mean excitation number $\nb_j:=\langle \hat{a}_j^\dagger\hat{a}_j\rangle$ can be written as
\begin{equation}
    \nb_j=|\alpha_j|^2+\mu_{jj}-\frac{1}{2}.
\end{equation}
As an example, a single-mode Gibbs state with mean excitation number $\bar n$ can be described by $\FirM=0$, $\mu=\bar{n}+\frac{1}{2}$, and $\nu=0$.

In this representation, a $J$-mode Gaussian unitary operations is described by a $2J$-dimensional vector $\mathbf{d}=\begin{pmatrix}
    \vec{\alpha}_\mathbf{d}^* & \vec{\alpha}_\mathbf{d}
\end{pmatrix}^\tp$ and a $2J\times 2J$ unitary matrix $\GU=\begin{pmatrix}
    \CU^* & \SU \\
    \SU^* & \CU
\end{pmatrix}$ satisfying $\GU\mmZ \GU^\dagger=\mmZ$, or equivalently,
\begin{equation}\label{eq:gau_u_con}
    \CU \CU^\dagger-(\SU\SU^\dagger)^* = \iden,\ (\SU\CU^\dagger)^\tp=\SU\CU^\dagger.
\end{equation}
which acts on the first and second moments as $\FirM\rightarrow \GU\FirM+\mathbf{d}$ and $\SecM\rightarrow\GU\SecM \GU^\dagger$. It is directly checked that $|\det(\GU)|=1$ and further $\det(\GU\SecM \GU^\dagger)=\det(\SecM)$, i.e., the determinant of the second moment does not change under the action of Gaussian unitary operations.

In the following, we will use the tuple $(\mathbf{d},\CU,\SU)$ to represent a multimode Gaussian unitary operation. For example, a displacement is described by $(\mathbf{d},\GU=\iden)$; a passive unitary, which is composed of phase-shifters and beam-splitters, is described by $\mathbf{d}=0$, $\SU=0$, and a unitary matrix $\CU$; a phase shifter further requires $\CU=\mathrm{diag}(e^{-i\phi_1},\dots,e^{-i\phi_J})$. For single-mode squeezing acting on each mode, $\mathbf{d}=0$, $\CU=\mathbf{\Lambda_C}=\mathrm{diag}(\cosh r_1,\dots,\cosh r_J)$, and $\SU=\mathbf{\Lambda_S}=\mathrm{diag}(\sinh r_1,\dots,\sinh r_J)$. With this, some known results can be inferred directly. For example, Eq.~(\ref{eq:gau_u_con}) implies that the singular value decompositions of $\CU=\mathbf{W}_\CU\mathbf{\Lambda}_\CU\mathbf{V}_\CU^\dagger$ and $\SU=\mathbf{W}_\SU\mathbf{\Lambda}_\SU\mathbf{V}_\SU^\dagger$ should satisfy $\mathbf{W}_\CU=\mathbf{W}_\SU^*$, $\mathbf{V}_\CU=\mathbf{V}_\SU$, and $\mathbf{\Lambda}_\CU^2-\mathbf{\Lambda}_\SU^2=\iden$. It follows that any $\GU$ in a Gaussian unitary can be decomposed as $\GU=\GU_{\mathrm{ps}_2}\GU_{\mathrm{sq}}\GU_{\mathrm{ps}_1}$, where $\GU_{\mathrm{sq}}=\begin{pmatrix}
    \mathbf{\Lambda}_\CU & \mathbf{\Lambda}_\SU \\
    \mathbf{\Lambda}_\SU & \mathbf{\Lambda}_\CU
\end{pmatrix}$ describes a single-mode squeezing, $\GU_{\mathrm{ps}_1}=\mathbf{V}_\CU^{\dagger*}\oplus \mathbf{V}_\CU^\dagger$ and $\GU_{\mathrm{ps}_2}=\mathbf{W}_\CU^{*}\oplus \mathbf{W}_\CU$ are passive unitary operators. This immediately leads to the well-known result that any Gaussian unitary is decomposable as passive unitaries, single-mode squeezings, and displacements.

Besides, consider a single-mode Gaussian state $\rho^G$ with first and second moments $(\alpha, \mu,\nu)$. It can be transformed to a Gibbs state $\tau(\rho^G)$ with $\mu'=\sqrt{\mu^2-|\nu|^2}=\sqrt{\det(\SecM)}$ by a single-mode Gaussian unitary composed of a displacement by $-\alpha$, a phase shifter with $\phi=-\frac{1}{2}\arg(\nu)$, and a squeezing with $r=-\frac{1}{2}\arg\tanh\frac{|\nu|}{\mu}$. It follows that $\nb(\rho^G) = |\alpha|+\mu-\frac{1}{2}\geq \sqrt{\det(\SecM)}-\frac{1}{2}=\nb[\tau(\rho^G)]$. Hence, the thermal excitation in $\rho^G$ is $\ThEx(\rho^G)=\nb[\tau(\rho^G)]$ from the definition.

The representation based on $\{\hat{a}_j,\hat{a}_j^\dagger\}$ described above will be employed to simplify the proofs of our results. It is equivalent to the more widely used representation based on $\{\hat{x}_j=\frac{1}{\sqrt{2}}(\hat{a}_j+\hat{a}_j^\dagger),\hat{p}_j=\frac{1}{i\sqrt{2}}(\hat{a}_j-\hat{a}_j^\dagger)\}$ by a simple transformation. 

Next we will focus on the ability of multimode Gaussian unitary operations in manipulating the mean number of excitations in each mode. Lemma~\ref{le:eigen_value} below will be used in our proof of Lemma~\ref{le:gau_u}.
\begin{lemma}\label{le:eigen_value}
    Let $\LU$ be a $J\times J$ matrix satisfying $\LU\LU^\dagger\geq\iden$, and $O$ be a completely positive semidefinite matrix of the same dimension. The eigenvalues of $O$ in the descending order are labeled as $o_j^\downarrow$, while the eigenvalues of $\LU O\LU^\dagger$ in the descending order are labeled as $o_j^{\prime\downarrow}$. Then for all $j$, the following relation holds
    \begin{equation}
        o_j^{\prime\downarrow}\geq o_j^\downarrow.
    \end{equation}
    Let $\{o''_j\}$ be the diagonal elements of $\LU O\LU^\dagger$, and then
    \begin{equation}\label{eq:diag_eigen}
        \sum_{j=k}^J o_j''\geq\sum_{j=k}^J o_j^\downarrow,\ k=1,\dots,J.
    \end{equation}
    \begin{proof}
        By the Courant–Fischer–Weyl min-max principle, the eigenvalues $\lambda_j^\downarrow$ (where $\downarrow$ indicates the descending order) of a Hermitian matrix $Q$ can be calculated as
        \begin{eqnarray}
            \lambda_j^\downarrow=\max_{\begin{subarray}{c}
                \mathcal{M}\in\mathcal{H} \\
                 \mathrm{dim}(\mathcal{M})=j
            \end{subarray}}\min_{\begin{subarray}{c}
                |x\rangle\in\mathcal{M} \\
                 \langle x|x\rangle=1
            \end{subarray}} \langle x|Q|x\rangle,
        \end{eqnarray}
        where $\mathcal{H}$ is the Hilbert space on which $Q$ acts. 
        
From $\LU\LU^\dagger\geq\iden$, we have $\|\LU^\dagger|x\rangle\|^2=\langle x|\LU\LU^\dagger|x\rangle\geq 1$ for any normalized vector $|x\rangle$. Defining $|x_\LU\rangle=\LU^\dagger|x\rangle/\|\LU^\dagger|x\rangle\|$, we then obtain
\begin{equation}\label{eq:norm_xL}
    \langle x|\LU O\LU^\dagger|x\rangle=\langle x_\LU |O|x_\LU\rangle\cdot\|\LU^\dagger|x\rangle\|^2\geq \langle x_\LU |O|x_\LU\rangle.
\end{equation}
Besides, $\LU$ is full-rank, so for any $j$-dimensional space $\mathcal{M}$, $\mathcal{M}_\LU \equiv\{ \LU^\dagger|x\rangle\big| |x\rangle\in \mathcal{M}\}$ is also $j$-dimensional. It follows that for any normalized vector $|x\rangle\in\mathcal{M}$, there is a normalized vector $|x_\LU\rangle\in\mathcal{M}_\LU$, such that Eq.~(\ref{eq:norm_xL}) holds. Therefore, for any $j$-dimensional space $\mathcal{M}$, there is a $j$-dimensional space $\mathcal{M}_\LU$ such that
\begin{equation}
    \min_{\begin{subarray}{c}
                |x\rangle\in\mathcal{M} \\
                 \langle x|x\rangle
                 \end{subarray}} \langle x|\LU O\LU^\dagger|x\rangle \geq \min_{\begin{subarray}{c}
                |x_\LU\rangle\in\mathcal{M}_\LU \\
                 \langle x_\LU|x_\LU\rangle
                 \end{subarray}} \langle x_\LU| O|x_\LU\rangle.
\end{equation}
It follows immediately that
\begin{equation}
    o_j^{\prime\downarrow}=\max_{\begin{subarray}{c}
                \mathcal{M}\in\mathcal{H} \\
                 \mathrm{dim}(\mathcal{M})=j
            \end{subarray}} \min_{\begin{subarray}{c}
                |x\rangle\in\mathcal{M} \\
                 \langle x|x\rangle
                 \end{subarray}} \langle x|\LU O\LU^\dagger|x\rangle 
                 \geq \max_{\begin{subarray}{c}
                \mathcal{M}_\LU\in\mathcal{H} \\
                 \mathrm{dim}(\mathcal{M}_\LU)=j
            \end{subarray}} \min_{\begin{subarray}{c}
                |x_\LU\rangle\in\mathcal{M}_\LU \\
                 \langle x_\LU|x_\LU\rangle
                 \end{subarray}} \langle x_\LU| O|x_\LU\rangle = o_j^\downarrow.
\end{equation}
Further, Eq.~(\ref{eq:diag_eigen}) is obtained by recalling that the diagonal elements $\{o''_j\}$ and the eigenvalues $\{o^{\prime\downarrow}_j\}$ of $\LU O\LU^\dagger$ are related as $\sum_{j=k}^J o''_j\geq \sum_{j=k}^Jo^{\prime\downarrow}_j$.

    \end{proof}
\end{lemma}

\begin{lemma}\label{le:gau_u}
    Consider a system consisting of $J$ independent harmonic oscillators with frequency $\omega_1,\dots,\omega_J$ and initially in the Gibbs state at inverse temperature $\beta$. The mean excitation number of the $j$th mode is $\bar{n}_j=\frac{1}{e^{\beta\omega_j}-1}$. Under any Gaussian unitary $(\mathbf{d},\CU,\SU)$, the mean excitation number $\bar{n}_l'$ in the output satisfies
    \begin{equation}
        \sum_{j=k}^{J}\bar{n}_j'\geq\sum_{j=k}^J \bar{n}_j^\downarrow,\forall k=1,\dots,J,\label{eq:le_gau_U}
    \end{equation}
    where $\{\bar{n}_j^\downarrow\}$ is $\{\bar{n}_j\}$ in the descending order.
    The equalities are saturated for all $k$ when $\mathbf{d}=0$, $\SU=0$ and $\CU$ is a permutation which puts $\bar{n}_j$ in the descending order.
\end{lemma}
This lemma is related to Lemma 1 in Ref.~\cite{PhysRevA.73.012330}, where the relation between the diagonal elements and the symplectic eigenvalues of a covariant matrix has been demonstrated. Here we are more interested in the conversion of the mean excitation numbers of each mode under the action of a multimode Gaussian unitary operation. For completeness, we provide our comparably simpler proof.
\begin{proof}
    Because the system is initially in the Gibbs state at inverse temperature $\beta$, the second moment of the $J$-mode initial state is in the following form
\begin{equation}
    \mu=\mathrm{diag}(\mu_1,\dots,\mu_J),\nu=0,
\end{equation}
where $\mu_j=\bar{n}_j+\frac{1}{2}$, and $\bar n_j=\frac{1}{e^{\beta\omega_j}-1}$ is the mean excitation number of the $j$th mode. 
After a joint Gaussian unitary $\GU$, the second moment becomes
\begin{eqnarray}
    \mu' & = & \CU\mu \CU^\dagger+\SU^*\mu \SU^{*\dagger},\\
    \nu' & = & \CU^*\mu \SU^{*\dagger}+\SU\mu \CU^\dagger.
\end{eqnarray}
The second moment of the $j$-th mode is then calculated as
\begin{eqnarray}
    \mu_j' & = & \sum_{j'} (|\CU_{jj'}|^2+|\SU_{jj'}|^2)\mu_{j'},\\
    \nu_j' & = & 2\sum_{j'} \CU_{jj'}^*\SU_{jj'}\mu_{j'}.
\end{eqnarray}
It follows that for all $k=1,\dots,J$,
\begin{eqnarray}
    \sum_{j=k}^{J}\mu_j'=\sum_{j=k}^{J}\sum_{j'=1}^J (|\CU_{jj'}|^2+|\SU_{jj'}|^2)\mu_{j'}
    \geq \sum_{j=k}^{J}\sum_{j'=1}^J |\CU_{jj'}|^2\mu_{j'}=\sum_{j=k}^J\mu_j'',
\end{eqnarray}
where $\{\mu_j''\}$ are the diagonal elements of the matrix $\CU\mu \CU^\dagger$. 
From Eq.~(\ref{eq:gau_u_con}), $\CU\CU^\dagger\geq\iden$.
It follows from Lemma~\ref{le:eigen_value} that
$\sum_{j=k}^{J}\mu_j''\geq \sum_{j=k}^J \mu_j^\downarrow$ for all $k=1,\dots,J$, and then $\sum_{j=k}^{J}\mu_j'\geq \sum_{j=k}^J \mu_j^\downarrow$. Recalling that $\bar n_j'\geq\mu_j'-\frac{1}{2}$ is the mean excitation number of the $j$th mode in the output, and $\nb_j^\downarrow=\mu_j^\downarrow-\frac{1}{2}$ is the mean excitation number in the input, we arrive at Eq.~\ref{eq:le_gau_U}.
It is noted that the left-hand side means the summation of the mean excitation numbers of arbitrary $k$ modes in the output.
\end{proof}

\subsection{The cooling limit of the Gaussian HBAC}\label{app:A2}
Two consequences of Lemma~\ref{le:gau_u} are observed. The first is Lemma~\ref{le:global_U} in the main text, which we restated below.

\begin{lemma}\label{le:global_U_app}
    Consider a $J$-mode system initially in a Gaussian product state $\rho_J=\otimes_{j=1}^J\rho_j$. 
    After the action of a global unitary, the minimum of the thermal excitation of Mode 1 in the output is
    \begin{equation}
        \min_{W\in\mathcal{U}^G_J}\ThEx\left(\tr_{\backslash1}(W\rho_J W^\dagger)\right)
        =\min_j \ThEx(\rho_j),
    \end{equation}
    where $\mathcal{U}^G_J$ is the set of $J$-mode Gaussian unitary operators.
\begin{proof}
From $\rho_j=W_j\tau(\rho_j)W_j^\dagger$, $\ThEx(\rho_j)=\bar{n}[\tau(\rho_j)]$, and the definition of the thermal excitation, we have for all $J$-mode Gaussian unitary $W$, it holds that
    \begin{eqnarray}
        \ThEx\left(\tr_{\backslash1}(W\rho_J W^\dagger)\right) & = & \min_{\tilde W_1\in\mathcal{U}^G_1}\nb\left(\tilde W _1\tr_{\backslash1}( W(\otimes_{j=1}^J W_j)( \otimes_{j=1}^J\tau(\rho_j))( \otimes_{j=1}^JW_j^\dagger) W^\dagger)\tilde W_1^\dagger\right) \nonumber\\
        & = & \nb \left[\tr_{\backslash1}\left(\bar{W}(\otimes_{j=1}^J\tau(\rho_j))\bar{W}^\dagger\right)\right] \geq \min_j\nb[\tau(\rho_j)]=\min_j\ThEx(\rho_j),
    \end{eqnarray}
    where $\bar{W}=\tilde{\bar W}_1 W (\otimes_{j=1}^J W_j)$ and $\tilde{\bar W}_1$ is the single-mode Gaussian unitary that saturate the minimization in the first line. The equality is saturated by choosing $W$ to be a swap between the first mode and the mode with the minimum thermal excitation. This completes the proof.
\end{proof}
\end{lemma}

\subsection{The most efficient protocols of Gaussian HBAC}\label{app:A3}
The second consequence of Lemma~\ref{le:gau_u} is the following. 
\begin{lemma}\label{le:max_K}
Consider a $J$-mode Gibbs state $\tau=\otimes_{j=1}^J\tau_j$ with the mean excitation number of each mode in the descending order $\{\bar{n}_j^\downarrow\}$, and a $J$-mode Gaussian state $\rho=U^G\tau U^{G\dagger}$ obtained from $\tau$ via a Gaussian unitary $U^G$ with the mean excitation numbers $\bar{n}_j^\rho =\tr(\rho \hat{a}_j^\dagger \hat{a}_j)$. Let $\vec{K}=(K_1,\dots,K_{J'})$ ($J'\leq J$) be a real vector with its elements in the increasing order. If $\bar{n}_1^\rho=\bar{n}_J^\downarrow,\dots,\bar{n}^\rho_{J-{J'}}=\bar{n}_{{J'}+1}^\downarrow$, then
    \begin{equation}
        \sum_{j=1}^{J'} K_j \nb_{J-{J'}+j}^\rho \geq \sum_{j=1}^{J'} K_j \nb_j^\downarrow.
    \end{equation}
    The equality is saturated when $\nb^\rho_{J-{J'}+j}=\nb_j^\downarrow$ for $j=1,\dots,{J'}$.
    \begin{proof}
        For ${J'}=J$, we have
        \begin{eqnarray}
            \sum_{j=1}^J K_j \nb_{j}^\rho & = & \sum_{j=1}^J K_j\left(\sum_{l=j}^J \nb^\rho_l-\sum_{l=j+1}^J \nb_l^\rho\right)\nonumber\\
            & = & K_1\sum_{l=1}^J \nb^\rho_l + \sum_{j=2}^J(K_j-K_{j-1})\sum_{l=j}^J \nb^\rho_l\nonumber\\
            & \geq & K_1\sum_{l=1}^J \nb^\downarrow_l + \sum_{j=2}^J(K_j-K_{j-1})\sum_{l=j}^J \nb^\downarrow_l\nonumber\\
            & = & \sum_{j=1}^J K_j \nb_{j}^\downarrow.
        \end{eqnarray}
        In the third line, we have used Lemma~\ref{le:gau_u}. The equality is saturated when $\nb_j^\rho=\nb_j^\downarrow$ for all $j$.

        For ${J'}<J$, define $\vec{K}'$ such that $K'_{j} =K_j$ for $j=1,\dots,{J'}$, and $K_{{J'}}'\leq K_{{J'}+1}'\leq\dots\leq K'_{J-1}\leq K'_{J}$. It follows that
        \begin{eqnarray}
            \sum_{j=1}^{J'} K_j \nb_{J-{J'}+j}^\rho & = &  \sum_{j=1}^{J'} K'_j \nb_{J-{J'}+j}^\rho + \sum_{j={J'}+1}^{J} K'_j \nb_{J-j+1}^\rho -\sum_{j={J'}+1}^{J} K'_j \nb_{J-j+1}^\rho\nonumber\\
            & = & \sum_{j=1}^J K'_j \nb_{j}^{\rho'} - \sum_{j={J'}+1}^{J} K'_j \nb_{j}^\downarrow\nonumber\\
            & \geq & \sum_{j=1}^J K'_j \nb_{j}^\downarrow - \sum_{j={J'}+1}^{J} K'_j \nb_{j}^\downarrow  =  \sum_{j=1}^{J'} K'_j \nb_{j}^\downarrow.        
        \end{eqnarray}
        where $\rho'$ is a Gaussian state related to $\rho$ as $\nb^{\rho'}_j=\nb_{J-{J'}+j}^\rho$ for $j=1,\dots,{J'}$ and $\nb^{\rho'}_j=\nb_{J-j+1}^\rho$ for $j={J'}+1,\dots,J$.
    \end{proof}
\end{lemma}

Now we are ready to present Theorem~\ref{thm:2} in the main text and its proof.
\begin{theorem}\label{thm:2_app}
    For any choice of $\omega_1,\dots,\omega_{N-1}$, among all the protocols which can achieve the cooling limit $\beta^*=\frac{\omega_N}{\omega_0}\beta$, the most efficient unitary in the recharging routine is
    \begin{eqnarray}
        V^*& := & \mathrm{arg}\min_{V\in\mathcal{U}^G_{N+1}}\Sigma_N(V,H_M)\nonumber\\
        & = & \Swap(S,M_N)\circ\dots\circ\Swap(S,M_{j_0}),
    \end{eqnarray}
    where $\Swap(\cdot,\cdot)$ is the unitary which swaps the states of two modes, $j_0$ is the minimum index such that $\omega_{j_0}>\omega_0$. Besides, the minimum entropy production reads
    \begin{equation}
        \Sigma_N^*(H_M)=D[\tau(\omega_0)\| \tau(\omega_{j_0})]+\sum_{j=j_0+1}^N D[\tau(\omega_{j-1}) \| \tau(\omega_j)],\label{eq:ent_pro_star_app}
    \end{equation}
    where $\tau(\omega_j)=(1-e^{-\beta\omega_j})\sum_{n=0}^\infty e^{-n\beta\omega_j}|n\rangle\langle n|$ is the Gibbs state of a harmonic oscillator with frequency $\omega_j$ and at inverse temperature $\beta$.
    \begin{proof}
        The initial states of the system and the machine are $\tau_S=\tau(\omega_0)$ and $\tau_M=\otimes_{j=1}^N\tau(\omega_j)$, respectively, and the final state of the system is $\rho_S'=\tau(\omega_N)$. With this, we have
        \begin{eqnarray}
            I_{S:M}(\rho'_{SM}) & = & S(\rho'_S)+S(\rho'_M)-S(\rho'_{SM})\nonumber\\
            & = & S(\tau(\omega_N))+S(\rho'_M)-S(\tau_S\otimes\tau_M)\nonumber\\
            & = & S(\rho'_M)-\sum_{j=0}^{N-1}S(\tau(\omega_j)).
        \end{eqnarray}

    It follows that
    \begin{eqnarray}\label{eq:sigma}
        \Sigma_N(U,H_M) & = & -\tr(\rho'_M \ln \tau_M)-S(\rho_M')+I_{S:M}(\rho'_{SM})\nonumber\\
        & = & -\tr(\rho'_M\ln(\otimes_{j=1}^N\tau(\omega_j)))-\sum_{j=0}^{N-1}S(\tau(\omega_j))\nonumber\\
        & = & -\sum_{j=1}^N\tr(\rho'_j\ln\tau(\omega_j))-\sum_{j=0}^{N-1}S(\tau(\omega_j)),
    \end{eqnarray}
    where $\rho'_j=\tr_{\backslash j}(\rho'_M)$. Let $\nb_j:=\tr(\tau(\omega_j)a_j^\dagger a_j)$ and $\nb_j'=\tr(\rho_j'a_j^\dagger a_j)$ be the mean excitation number of $\tau(\omega_j)$ and $\rho'_j$, respectively. Recalling that $\tau(\omega_j)=\sum_{l=0}^\infty\frac{\nb_j^l}{(\nb_j+1)^{l+1}}|l\rangle\langle l|$, we have
    \begin{eqnarray}\label{eq:sigma1}
        -\sum_{j=1}^N\tr(\rho'_j\ln\tau(\omega_j)) & = & \sum_{j=1}^N\sum_{l=0}^\infty \langle l|\rho_j'|l\rangle [(l+1)\ln(\nb_j+1)-l\ln \nb_j]\nonumber\\
        & = & \sum_{j=1}^N\sum_{l=0}^\infty l\langle l|\rho_j'|l\rangle\ln\frac{\nb_j+1}{\nb_j}+\sum_{j=1}^N\ln (\nb_j+1)\nonumber\\
        & = & \sum_{j=1}^N\nb'_j\ln\frac{\nb_j+1}{\nb_j}+\sum_{j=1}^N\ln(\nb_j+1)\nonumber\\
        & = & \beta\sum_{j=1}^N\nb'_j\omega_j+\sum_{j=1}^N\ln(\nb_j+1).
    \end{eqnarray}
    By setting $J=N+1$, $J'=N$, $\vec{K}=\{\omega_1,\dots,\omega_{N}\}$, $\{\bar{n}^\rho_j\}=\{\bar{n}'_0,\bar{n}'_1,\dots,\bar{n}_N'\}$, and $\{\bar{n}^\downarrow_j\}=\{\bar{n}_1,\dots,\bar{n}_{j_0-1},\bar{n}_0,\bar{n}_{j_0},\dots,\bar{n}_N\}$ in Lemma~\ref{le:max_K}, because $\nb_0'=\nb_N$, we have
    \begin{equation}\label{eq:sigma2}
        \sum_{j=1}^N\nb'_j\omega_j\geq \sum_{j=1}^{{j_0}-1} \nb_j\omega_j+\nb_0\omega_k+\sum_{j={j_0}+1}^N \nb_{j-1}\omega_j,
    \end{equation}
    where the equality is saturated when the system successively swap its state with the $j$th mode in the machine ($j={j_0},\dots,N$). Combining Eqs.~(\ref{eq:sigma}), (\ref{eq:sigma1}) and (\ref{eq:sigma2}) together, we obtain
    \begin{eqnarray}
        \Sigma_N^*(H_M) & = & \beta\left[\sum_{j=1}^{{j_0}-1} \nb_j\omega_j+\nb_0\omega_k+\sum_{j={j_0}+1}^N \nb_{j-1}\omega_j\right]+\sum_{j=1}^N\ln(\nb_j+1)-\sum_{j=0}^{N-1}S(\tau(\omega_j))\nonumber\\
        & = & \beta\left[\nb_0\omega_k+\sum_{j={j_0}+1}^N \nb_{j-1}\omega_j\right]+\sum_{j={j_0}}^N\ln(\nb_j+1)-S(\tau(\omega_0))-\sum_{j={j_0}+1}^{N}S(\tau(\omega_{j-1}))\nonumber\\
        & = & D[\tau(\omega_0)||\tau(\omega_k)]+\sum_{j={j_0}+1}^N D[\tau(\omega_{j-1})||\tau(\omega_j)].
    \end{eqnarray}

\end{proof}
\end{theorem}
\begin{figure}[!htbp]
    \centering
    (a)\includegraphics[width=0.4\textwidth]{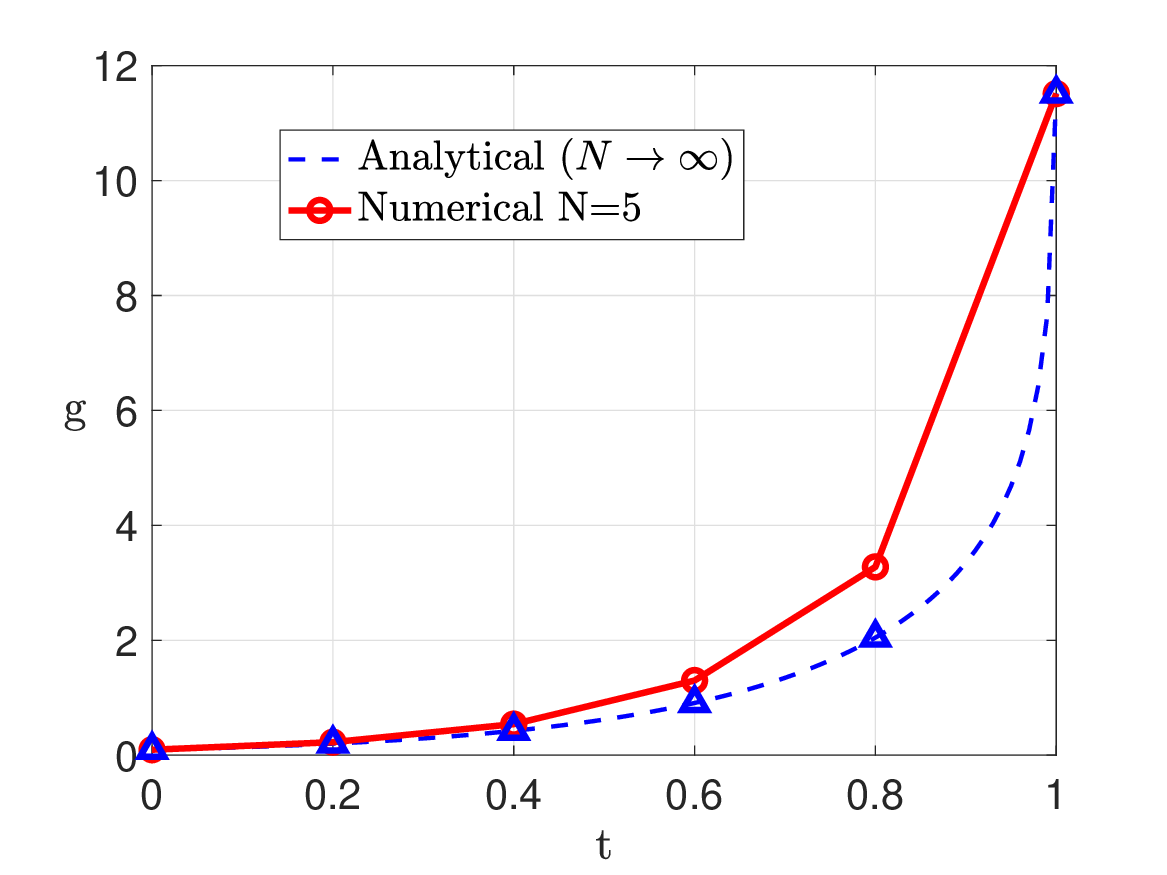}
    (b)\includegraphics[width=0.4\textwidth]{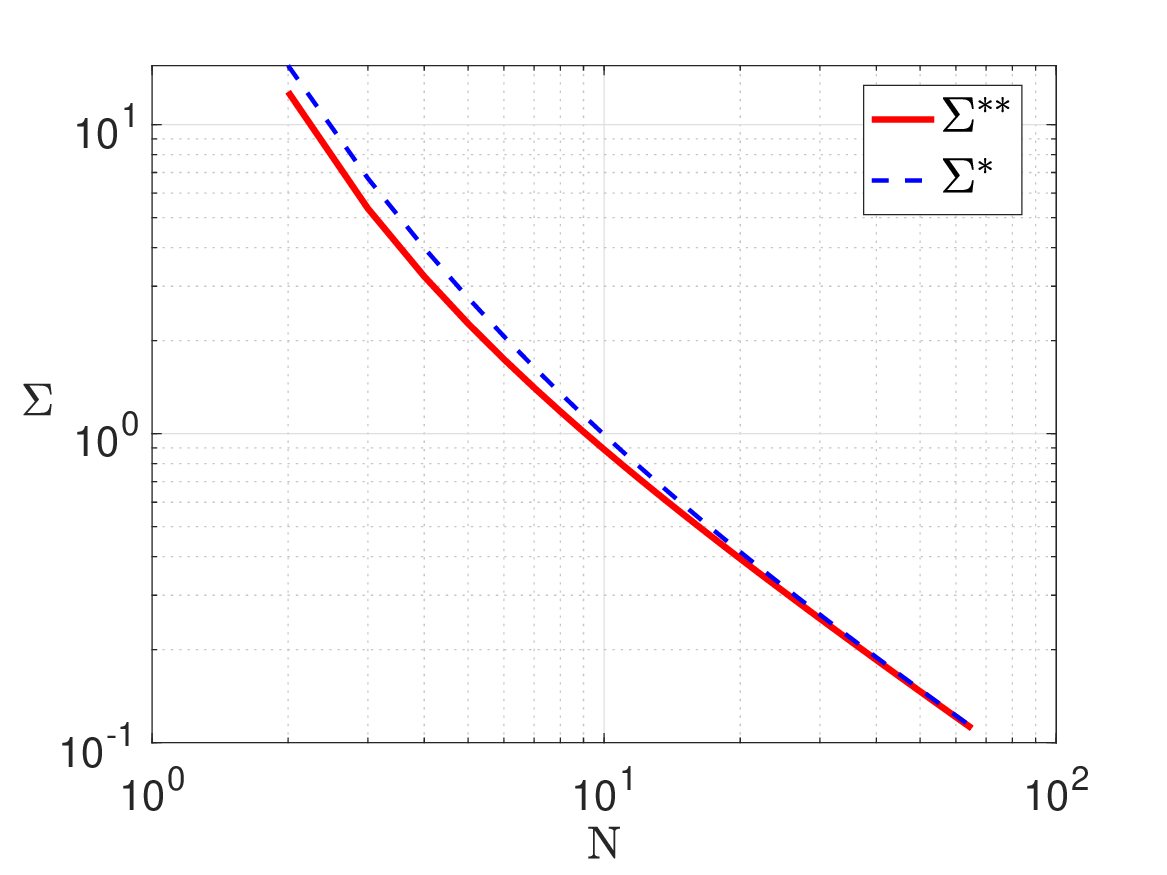}
    \caption{ (a) Trajectory $g$ vs $t$ for cooling from $\bar{n}_s(0) =10$ to $\bar{n}_s(1) =10^{-5}$, which corresponds to  $\lambda = 120.8 $. The red markers represent optimized energy gap of the machine for five timesteps. The blue markers are the discrete solution sampled from the analytical result. (b) Entropy production vs $N$, where $\Sigma_N^{**} (\Sigma_N^{*})$ is the entropy production with numerical solution (discrete timestep solution sampled from the analytical curves) }
    \label{fig:Traj}
\end{figure}
In the limit where $N$ is large, such that $(g_j-g_{j-1})^k N^2$ is infinitesimal for $k\geq3$, from Eq.~(\ref{eq:opt_spec}), we have the following
\begin{eqnarray}
    \frac{\frac{g_{j+1}-g_j}{1/N}-\frac{g_{j}-g_{j-1}}{1/N}}{1/N} & = & \left(\frac{1}{1/N}\right)^2\left[(e^{g_j-g_{j-1}}-1)\cdot\frac{1-e^{-g_{j}}}{1-e^{-g_{j-1}}}-(g_j-g_{j-1})\right]\nonumber\\
    & = & \left(\frac{1}{1/N}\right)^2\left[(g_j-g_{j-1})\cdot\left(\frac{1-e^{-g_{j}}}{1-e^{-g_{j-1}}}-1\right)+\left(\frac{1}{2}(g_j-g_{j-1})^2+\frac{1}{3!}(g_j-g_{j-1})^3+\dots\right)\cdot\frac{1-e^{-g_{j}}}{1-e^{-g_{j-1}}}\right]\nonumber\\
    & = & \left(\frac{1}{1/N}\right)^2\left[(g_j-g_{j-1})\cdot\frac{1-e^{g_{j-1}-g_j}}{e^{g_{j-1}}-1}+\left(\frac{1}{2}(g_j-g_{j-1})^2+\frac{1}{3!}(g_j-g_{j-1})^3+\dots\right)\cdot\frac{e^{g_{j-1}}-e^{g_{j-1}-g_{j}}}{e^{g_{j-1}}-1}\right]\nonumber\\
    & = & \left(\frac{g_j-g_{j-1}}{1/N}\right)^2\cdot\frac{e^{g_{j-1}}+1}{2(e^{g_{j-1}}-1)}+o[N^2(g_j-g_{j-1})^3].
\end{eqnarray}
When $g_j$ is approximately continuous, let $t=j/N$, $dt=1/N$ and $g_j=g(t)$. The above equality reduces to
\begin{equation}
    \ddot{g}=\dot{g}^2\cdot\frac{e^{g_{j-1}}+1}{2(e^{g_{j-1}}-1)}.\label{eq:opt_tra_eq}
\end{equation}
This equation can be easily solved, and we obtain the following optimal spectrum of the machine
\begin{equation}
    g^*_t=2\ln\left\{-\coth\left[\frac{t}{2}\ln\tanh\frac{g_N}{4}+\frac{1-t}{2}\ln\tanh\frac{g_0}{4}\right]\right\},
\end{equation}
or in the discrete form
\begin{equation}
    g^*_j=2\ln\left\{-\coth\left[\frac{j}{2N}\ln\tanh\frac{g_N}{4}+\frac{N-j}{2N}\ln\tanh\frac{g_0}{4}\right]\right\}.
\end{equation}
Therefore, the optimal entropy production for large $N$ reads
\begin{equation}
    \Sigma^*_{N\gg1}=\frac{1}{2N}\left[\frac{1}{2}\ln\frac{\tanh(\beta\omega_N/4)}{\tanh(\beta\omega_0/4)}\right]^2+o(\frac{1}{N^2}). \label{eq:solDiscrete}
\end{equation}

 We find that the numerical solutions of the cooling trajectory and the corresponding entropy production for finite $N$ begin to deviate from the large-$N$ limit when targeting very low mean excitations, on the order of $10^{-5}$ (see Fig.  \ref{fig:Traj}). 
For moderate cooling targets with mean excitations around $10^{-2}$, the deviation from the large-$N$ behavior becomes negligible. In such cases, one can get minimized entropy production by using solution that sampled from $N$ equally spaced points along the analytical trajectory given in Eq.~\ref{eq:solDiscrete}. 

\section{Appendixes for Non-Gaussian cooling}

\subsection{Evolution of mean excitation number under Heisenberg picture }\label{app:B1}

 The mean excitation number of target after single operation of $V^{(p)}(t)$ is given by

 \begin{eqnarray}
         \bar{n}_S(t) &= &\tr\left[ V^{(p)}(t)(\rho_S \otimes \rho_M) V^{(p) \dagger}(t) (\hat{n}_S \otimes \iden) \right ]\nonumber\\
         & =& \tr\left[ \mathcal{E}_t(\rho_S) \hat{n}_S \right ]
 \end{eqnarray}
where $\mathcal{E}_t(\rho_S) = \tr_M[V^{(p)}(t)(\rho_S \otimes \rho_M) V^{(p) \dagger}(t)] $. Equivalently in Heisenberg picture, we have

\begin{align}
    \bar{n}_S(t) &= \tr\left[ \tau_S  \mathcal{E}_t^\dag (\hat{n}_S ) \right ]
\end{align}

where $\mathcal{E}_t^\dag(\hat{n}_S) = \tr_M \left[ (\iden \otimes \sqrt{\tau_M}) V^{(p) \dag}(t) (\hat{n}_S \otimes \iden ) V^{(p)}(t) (\iden \otimes \sqrt{\tau_M}) \right]$. In short-time regime $\chi t\ll 1$, the mean excitation number is approximate using the Baker-Campbell Haussdorf  (BCH) expansion 
\begin{equation}
    V^{(p) \dag}(t) (\hat{n}_S \otimes \iden ) V^{(p)}(t)  \approx \bar{n}_S + i [H^{(p)},\hat{n}_S]t-\frac{1}{2} [H^{(p)},[H^{(p)},\hat{n}_S]] t^2
\end{equation}
where
\begin{equation}
    [H^{(p)},\hat{n}_S]= \chi (\hat{a} b^{\dag p} - \hat{a}^\dag b^{p} )\label{eq:2ndOrder}
\end{equation}
\begin{equation}
   [H^{(p)},[H^{(p)},\hat{n}_S]] = 2 \chi^2 \left[\hat{n}_S [\hat{b}^p , \hat{b}^{\dag p}] - \hat{b}^{\dag p} \hat{b}^p  \right] + \chi (\omega_0 - p \omega_1) (\hat{a} \hat{b}^{\dag p} + \hat{a}^\dag \hat{b}^p  ) \label{eq:3rdOrder}
\end{equation}


If the initial states of the system and machine are assumed to to Gibbs state $\tau_S$, $\tau_M$  with mean excitation number $\bar{n}_S$ and $ \bar{n}_M$ respectively, the evolution of the mean excitation number of target $\bar{n}_{S}(t)$ and machine $\bar{n}_{M}(t)$ are approximate to (in the small time perturbation).

\begin{eqnarray}
    \bar{n}_{S}(t) &= & \bar{n}_S - \left[ \bar{n}_S \langle[\hat{b}^p , \hat{b}^{\dag p}] \rangle_{\tau_M} - \langle \hat{b}^{\dag p} \hat{b}^p  \rangle_{\tau_M} \right] \chi ^2 t^2 + \mathcal{O}(t^4) +...\nonumber\\
    &= & \bar{n}_S - \left[ (1+\bar{n}_M)^p \bar{n}_S  - \bar{n}_M^p (1+\bar{n}_S )\right] p! \chi ^2 t^2 + \mathcal{O}(t^4) +...
\end{eqnarray}

where we have used
\begin{eqnarray}
    \sum_{n=0}^{\infty}\frac{\bar{n}^{n}}{(\bar{n}+1)^{n+1}}\langle \hat{b}^{\dag p} \hat{b}^p  \rangle_n &=& \sum_{n=0}^{\infty}\frac{\bar{n}^{n}}{(\bar{n}+1)^{n+1}}\frac{n!}{(n-p)!}\nonumber\\
    &= & p! \bar{n}^p 
\end{eqnarray}

\begin{eqnarray}
    \sum_{n=0}^{\infty}\frac{\bar{n}^{n}}{(\bar{n}+1)^{n+1}}\langle \hat{b}^p  \hat{b}^{\dag p}  \rangle_n &=& \sum_{n=0}^{\infty}\frac{\bar{n}^{n}}{(\bar{n}+1)^{n+1}}\frac{(n+p)!}{(n!)}\nonumber\\
    &=& p! (1+\bar{n})^p
\end{eqnarray}

Using $tr[\tau_{S,M} a^m b^{\dag n }]= 0$, the contributions from Eq.~\eqref{eq:2ndOrder} and the  second term in Eq.~\eqref{eq:3rdOrder} vanish. Thus, $\bar{n}_{S}(t) $ is also valid for nonresonance case $\omega_0 \neq p\omega_1 $ if (i) the initial state $\tau$ is state that is diagonal in the Fock basis. (ii) short-time perturbation limit. 
Following similar calculations, we show that the first moment $\alpha(t)=\tr(\tau_S\mathcal{E}_t^\dagger(\hat{a}))=0$ and the off-diagonal element of the second moment $\nu(t)=\tr(\tau_S\mathcal{E}_t^\dagger(\hat{a}^2))=0$. Therefore, the thermal excitation of the output state $\rho_S'(t)=\mathcal{E}_t(\tau_S)$ equals to its mean excitation number, i.e., $\ThEx(\rho'_S(t))=\bar{n}_S(t)$.

\subsection{Cooling limit for iterative cooling (matrix algebra) }\label{app:B2}
The cooling protocol in the previous section is then repeated with $L$ iteration. In each iteration, the target is interact with reset qubit $\tau_S \otimes \tau_M$ with fixed machine gap $\omega_S,\omega_M $ and interaction duration $t$.  The mean photon number after $L$ iterations is given by 
 \begin{align}
         \bar{n}_S^{(L)}  = \tr\left[ \underbrace{\mathcal{E}_{t} \circ ... \circ \mathcal{E}_{t} }_{L \text{times} } (\rho_S) \hat{n}_S  \right ].
 \end{align}
Equivalently in the Heisenberg picture, we have
\begin{align}
    \bar{n}_S^{(L)} 
    &= \tr\left[ \rho_S  \underbrace{\mathcal{E}_{t}^\dag \circ ... \circ \mathcal{E}_{t}^\dag }_{L \text{times} } (\hat{n}_S ) \right ], \label{eq:HeisenbergPic}
\end{align}
and further,
\begin{equation}
    \mathcal{E}_{t}^\dag (\hat{n}_S) \approx \hat{n}_S - g^2  t^2 \left[\hat{n}_S (4 \bar{n}_M + 2) - 2 \iden \bar{n}_M^2  \right] + \mathcal{O}(t^4) .
\end{equation}
Then for $L$ iterations, we have
\begin{equation}
    \underbrace{\mathcal{E}_{t}^\dag \circ ... \circ \mathcal{E}_{t}^\dag }_{L \text{ times} } (\hat{n}_s)  = \underbrace{\mathcal{E}_{t}^\dag \circ ... \circ \mathcal{E}_{t}^\dag }_{(L -1) \text{ times}  }(\mathcal{E}_{t}^\dag (\hat{n}_s)), 
\end{equation}
and hence,
\begin{equation}
    \underbrace{\mathcal{E}_{t}^\dag \circ ... \circ \mathcal{E}_{t}^\dag }_{L \text{times} } (\hat{n}_S) \approx \underbrace{\mathcal{E}_{t}^\dag \circ ... \circ \mathcal{E}_{t}^\dag }_{(L -1) \text{ times}  }  (\hat{n}_S) - g^2  t^2 \left[ \underbrace{\mathcal{E}_{t}^\dag \circ ... \circ \mathcal{E}_{t}^\dag }_{(L -1) \text{ times}  }  (\hat{n}_S) (4 \bar{n}_M + 2) - 2 \iden \bar{n}_M^2  \right] + \mathcal{O}(t^4) .
\end{equation}
Plugging back into Eq. \ref{eq:HeisenbergPic}, we get the iterating function
\begin{equation}
\bar{n}_S^{(L)} \approx \bar{n}_S^{(L-1)} - g^2  t^2 \left[ \bar{n}_S^{(L-1)}  (4 \bar{n}_M + 2) - 2  \bar{n}_M^2  \right] + \mathcal{O}(t^4) .\label{eq:n_s_M_iterating}
\end{equation}
The closed form of the mean excitation number in term of initial mean excitation after $L$ iterations of cooling is as follows
\begin{equation}
    \bar{n}_S^{(L)} \approx \bar{n}_S^{(0)}(1-a)^L + b \frac{1-(1-a)^L}{a}
\end{equation}
where $a = \chi^2 t^2[(1+\bar{n}_M)^p + \bar{n}_M^p]p!$ and $b = \chi^2 t^2 \bar{n}_M^p p!$. 

\subsection{Thermality of the system after many iterations }\label{app:B3}

It is apparent that the system evolves into a transient state after a single collision which is not Gaussian nor thermal. In previous section, we found that the the mean excitation of the system converge to a thermal excitation. This raises a question whether this is coincidence or the final state actually equilibrates to a Gibbs state. To verify that, we use Fano factor $q(t)$ defined as (Eq.~9 in \cite{nimmrichter2017quantum})
\begin{equation}
    q(L) = \frac{\langle \hat{n}_S^2\rangle_L - \langle \hat{n}_S\rangle_L^2}{\langle \hat{n}_S\rangle_L(\langle \hat{n}_S\rangle_L +1)} -1
\end{equation}
where $\langle\bullet\rangle_L = \tr (\tau_S \underbrace{\mathcal{E}_{t}^\dag \circ ... \circ \mathcal{E}_{t}^\dag }_{L \text{times}} (\bullet))$. If $q(t) = 0$ then the system state has Gibbs distribution at the $L$th iteration . In the following calculation, we are interested in whether $\lim_{L\rightarrow \infty} q(t) = 0$, where 
\begin{equation}
    \lim_{L\rightarrow \infty} q(L) = \frac{\langle \hat{n}_S^2\rangle_\infty - \langle \hat{n}_S\rangle_\infty^2}{\langle \hat{n}_S\rangle_\infty(\langle \hat{n}_S\rangle_\infty +1)} -1, \label{eq:FanoInf}
\end{equation}
The analytical form of $\langle \hat{n}_S^2\rangle_1 $ and  $ \langle \hat{n}_S\rangle_1$ are given by
\begin{equation}
    \langle \hat{n}_S\rangle_1 =(1-a) \bar{n}_S^{(0)} + c,
\end{equation}
and
\begin{equation}
    \langle \hat{n}_S^2\rangle_1 =(1-2a) \bar{n^2}_S^{(0)} + b \bar{n}_S + c,
\end{equation}
where $\bar{n}_S^{(0)} = \langle \hat{n}_S\rangle_{\tau_S} $ and $\bar{n^2}_S^{(0)} = \langle \hat{n}_S^2\rangle_{\tau_S}$. Here, $a = \chi^2 t^2[(1+\bar{n}_M)^p - \bar{n}_M^p]p!$, $b=\chi^2 t^2[(1+\bar{n}_M)^p + 3\bar{n}_M^p]p!$ and $c = \chi^2 t^2 \bar{n}_M^p p!$. After $L$ iterations, the $\langle \hat{n}_S^2\rangle_L $ and  $ \langle \hat{n}_S\rangle_L$ in term of initial mean excitation read
\begin{equation}
    \bar{n}_S^{(L)} \approx \bar{n}_S^{(0)}(1-a)^L + c \frac{1-(1-a)^L}{a}
\end{equation}
and
\begin{align}
    \bar{n^2}_S^{(L)}
    &= (1 - 2a)^L \bar{n^2}_S^{(0)}
    + \bar{n}_S^{(0)} b \frac{1-a}{a} \Big[(1 - a)^L - (1 - 2a)^L\Big] \\
    &\quad + \frac{b c}{a}
    \left(
    \frac{1 - (1 - 2a)^L}{2a}
    - \frac{(1 - a)^L - (1 - 2a)^L}{a}
    \right)
    + c\,\frac{1 - (1 - 2a)^L}{2a}.
\end{align}
When $L \rightarrow \infty$ and $a\ll1$ (as implies by $\chi t \ll 1$), we have 
\begin{equation}
    \langle \hat{n}_S\rangle_\infty = \bar{n}_S^{(\infty)} \approx c/a
\end{equation}
\begin{equation}
    \langle \hat{n}_S^2\rangle_\infty=\bar{n^2}_S^{(\infty)} \approx \frac{c}{2a} \left( 1+ \frac{b}{c} \right)
\end{equation}
Plugging this into the Eq. \ref{eq:FanoInf} and using the relation  $1+b/c = 2 (1+ 2 c/a)$, we have

\begin{equation}
    \lim_{L\rightarrow \infty} q(L) = 0, \forall p.
\end{equation}
Therefore, when the cooling condition for $p$-excitation exchange is satisfied, its collision in larger $L$ limit equilibrates toward a Gibbs state for all nonlinearity order $p$. We verify the result with numerical example with $p=2,3$ (see Fig. \ref{fig:Gibbsconverge} ), where one can see that  the transient state (finite $L$) does not exhibit Gibbs distribution; however, in large $L$ limit the equilibrated state exhibits Gibbs distribution. Effectively, the collisional model with $p$-excitation exchange map a Gibbs state to Gibbs state, while in between the evolution can be out-of equilibrium, 

\begin{equation}
    \underbrace{\mathcal{E}_{t} \circ ... \circ \mathcal{E}_{t} }_{L \rightarrow \infty } (\tau_S (H_S,T_S)) = \tau_f (H_f,T_f),
\end{equation}
where $H_f$ is also a quadratic Hamiltonian, i.e., a harmonic oscillator.


\begin{figure}
    \centering
    \includegraphics[width=0.5\linewidth]{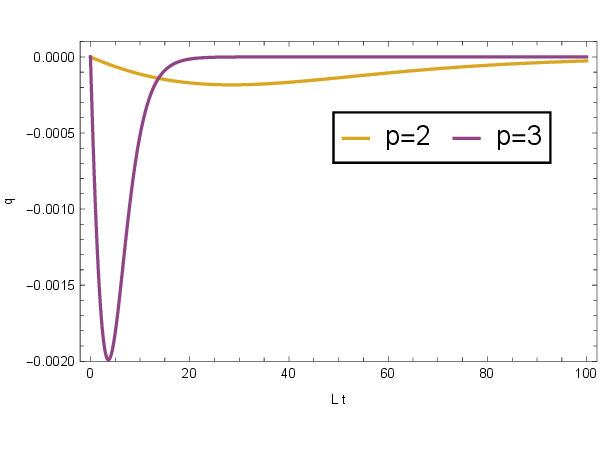}
    \caption{Plot of $q(L)$ vs $Lt$ for $p=2,3$ for iterative cooling with $L=1,2,3,...,2\times 10^4$ and $t=5\times10^{-3} s$ and initial mean photon number of  $\bar{n}_S^{(0)} = 2$, and $\bar{n}_M^{(0)}=1.5$.}
    \label{fig:Gibbsconverge}
\end{figure}

\end{document}